\documentclass[11pt,a4paper]{article}
\pdfoutput=1 

\usepackage{amssymb}
\usepackage{amsmath}
\usepackage{amsfonts}
\usepackage{dsfont}
\usepackage{amsthm}
\usepackage{mathrsfs}
\usepackage{mathtools}
\usepackage{hyperref}
\usepackage{color}
\usepackage[margin=2.41cm]{geometry}
\usepackage[all,cmtip]{xy}
\usepackage{graphicx}
\usepackage{varwidth}
\usepackage{comment}

\usepackage{expl3}

\usepackage[final]{microtype}

\usepackage{upgreek}
\usepackage{rotating}

\usepackage{tikz}
\usetikzlibrary{shapes.geometric}

\usepackage{tikz-cd}
\usetikzlibrary{matrix,arrows,calc,decorations.pathmorphing,fit,shapes.geometric,decorations.pathreplacing,positioning}

\makeatletter
\def\labelbox#1{%
  \hbox{%
    \setbox\z@=\hbox{$\m@th\labelstyle{\,#1\,}$}%
    \setbox\tw@=\hbox{$\m@th\labelstyle\,$}%
    \dimen@=\ht\z@ \advance\dimen@ by \wd\tw@ \ht\z@=\dimen@
    \dimen@=\dp\z@ \advance\dimen@ by \wd\tw@ \dp\z@=\dimen@
    \box\z@
  }%
}
\makeatother

\usepackage[shortlabels]{enumitem}

\renewcommand{\emptyset}{\varnothing}

\usepackage[style=alphabetic,
			date=year,
			backref=false,
			doi=true,
			url=false,
			giveninits=true,
			maxnames=4,
			maxalphanames=4,
			minalphanames=3]{biblatex}

\DeclareFieldFormat[article]{volume}{\mkbibbold{#1}}

\DeclareFieldFormat[book]{number}{\mkbibbold{#1}}
\DeclareFieldFormat[incollection]{number}{\mkbibbold{#1}}

\renewbibmacro{in:}{%
  \ifentrytype{article}{}{\printtext{\bibstring{in}\intitlepunct}}}

\DeclareFieldFormat[article]{pages}{#1}

\newcommand*{\clearisbn}{%
  \iffieldundef{doi}
    {}
    {\clearfield{isbn}\clearfield{issn}}}
\AtEveryBibitem{\clearisbn}
\AtEveryCitekey{\clearisbn}

\definecolor{darkred}{rgb}{0.8,0.1,0.1}
\hypersetup{
     colorlinks=false,         
     linkcolor=darkred,
     citecolor=blue,
}

\theoremstyle{plain}
\newtheorem{theo}{Theorem}[section]
\newtheorem{lem}[theo]{Lemma}
\newtheorem{propo}[theo]{Proposition}
\newtheorem{cor}[theo]{Corollary}

\theoremstyle{definition}
\newtheorem{defi}[theo]{Definition}

\newtheorem{assu}[theo]{Assumption}

\newtheorem{exx}[theo]{Example}
\newenvironment{ex}
  {\pushQED{\qed}\exx}
  {\popQED\endexx}

\newtheorem{remm}[theo]{Remark}
\newenvironment{rem}
  {\pushQED{\qed}\remm}
  {\popQED\endremm}

\numberwithin{equation}{section}

\def\nn{\nonumber}

\def\bbK{\mathbb{K}}
\def\bbR{\mathbb{R}}
\def\bbC{\mathbb{C}}

\def\bbZ{\mathbb{Z}}

\def\bbS{\mathbb{S}}

\def\ii{{\,{\rm i}\,}}

\def\Hom{\mathrm{Hom}}

\def\Imm{\mathrm{Im}}
\def\Ker{\mathrm{Ker}}

\def\id{\mathrm{id}}

\def\supp{\mathrm{supp}}

\def\dd{\mathrm{d}}
\def\vol{\mathrm{vol}}

\def\cc{\mathrm{c}}

\def\vc{\mathrm{vc}}
\def\dim{\mathrm{dim}}
\def\1{\mathbf{1}}
\def\oone{\mathds{1}}
\def\op{\mathrm{op}}
\def\ev{\mathrm{ev}}

\def\pr{\mathrm{pr}}

\def\2AQFT{\mathbf{2AQFT}}

\def\Open{\mathbf{Open}}

\def\Ch{\mathbf{Ch}}

\def\dgastAlg{\mathbf{dg}^\ast\mathbf{Alg}_\bbC}

\def\ext{\operatorname{ext}}

\def\AAA{\mathfrak{A}}
\def\BBB{\mathfrak{B}}
\def\LLL{\mathfrak{L}}

\def\BBB{\mathfrak{B}}

\def\CCC{\mathfrak{C}}

\def\FFF{\mathfrak{F}}

\def\Sol{\mathfrak{Sol}}
\def\CCR{\mathfrak{CCR}}

\def\se{{\scriptscriptstyle \searrow}}
\def\sw{{\scriptscriptstyle \swarrow}}
\def\ne{{\scriptscriptstyle \nearrow}}
\def\nw{{\scriptscriptstyle \nwarrow}}

\def\sk{\vspace{1mm}}

\makeatletter
\let\@fnsymbol\@alph
\makeatother

\title{%
The linear CS/WZW bulk/boundary system in AQFT
}

\author{%
Marco Benini$^{1,2,a}$, Alastair Grant-Stuart$^{3,b}$\ and\ 
Alexander Schenkel$^{3,c}$\vspace{4mm}\\
{\small ${}^1$ Dipartimento di Matematica, Universit\`a di Genova,}\\
{\small Via Dodecaneso 35, 16146 Genova, Italy.}\vspace{2mm}\\
{\small ${}^2$ INFN, Sezione di Genova,}\\
{\small Via Dodecaneso 33, 16146 Genova, Italy.}\vspace{2mm}\\
{\small ${}^3$ School of Mathematical Sciences, University of Nottingham,}\\
{\small University Park, Nottingham NG7 2RD, United Kingdom.}\vspace{4mm}\\
{\small \begin{tabular}{ll}
Email: & ${}^a$~\texttt{marco.benini@unige.it}\\
& ${}^b$~\texttt{alastair.grant-stuart@nottingham.ac.uk}\\
& ${}^c$~\texttt{alexander.schenkel@nottingham.ac.uk}\vspace{2mm}
\end{tabular}
}
}

\date{June 2023}

\begin{document}

\maketitle

\begin{abstract}
\noindent This paper constructs in the framework of algebraic quantum field theory (AQFT) the linear Chern-Simons/Wess-Zumino-Witten system on a class of $3$-manifolds $M$ whose boundary $\partial M$ is endowed with a Lorentzian metric. It is proven that this AQFT is equivalent to a dimensionally reduced AQFT on a $2$-dimensional manifold $B$, whose restriction to the $1$-dimensional boundary $\partial B$ is weakly equivalent to a chiral free boson.
\end{abstract}



\renewcommand{\baselinestretch}{0.8}\normalsize
\tableofcontents
\renewcommand{\baselinestretch}{1.0}\normalsize

\newpage


\section{\label{sec:intro}Introduction and summary}
Quantum field theories (QFTs) that are defined on manifolds with boundary
are known to exhibit interesting phenomena that arise from a rich
interplay between bulk and boundary degrees of freedom. A prime example
is given by the Chern-Simons/Wess-Zumino-Witten (in short, CS/WZW) bulk/boundary
system that describes a topological QFT in the bulk, a chiral conformal QFT
on the boundary and their non-trivial interplay, leading to interesting 
phenomena such as holography.
This was observed first
in \cite{CSWZW1} and subsequently studied in more detail by many authors, 
see e.g.\ \cite{CSWZW2,CSWZW3,CSWZW4,Mnev,PulmannSeveraValach2021,GRW,CattaneoMnevWernli2023}. One of the main
mechanisms causing an interplay between bulk and boundary degrees
of freedom is gauge invariance of the combined bulk/boundary system.
This leads to various consistency conditions between bulk and boundary fields 
that can be analyzed, for instance, in the BV/BFV formalism of 
Cattaneo, Mnev and Reshetikhin \cite{CMR1,CMR2}.
\sk

The aim of this paper is to study in detail a simple and explicit example
of such gauge-theoretic bulk/boundary system from the point of view
of algebraic QFT (AQFT) \cite{Haag,BFV}, or more precisely
its homotopical refinement \cite{BSWhomotopy,BSreview} that is
better suited to deal with gauge theories. The model
we study is given by (a Lorentz geometric variant of) the linear CS/WZW bulk/boundary system,
which describes linear Chern-Simons theory in the bulk of a $3$-manifold $M$ and 
a chiral free boson on its boundary $\partial M$, together with their interplay.
More precisely, we study Chern-Simons theory with linear 
structure group $G=\mathbb{R}$ on an oriented
$3$-manifold $M$, whose boundary $\partial M$ is endowed with (the conformal
class of) a Lorentzian metric and a time-orientation, in the presence
of a Lorentzian version of the chiral WZW boundary condition on $\partial M$.
A holomorphic variant of this example, where $\partial M$ is endowed
with a complex structure, was studied recently in the context of factorization algebras 
({\`a} la Costello-Gwilliam \cite{CostelloGwilliam,CostelloGwilliam2})
by Gwilliam, Rabinovich and Williams \cite{GRW}. 
\sk

The physics of the CS/WZW bulk/boundary system is already well-understood,
and by restricting to the linear structure group $G=\mathbb{R}$ we only 
study a particularly simple instance of it. Of interest in this work, though,
are the mathematical techniques we use to construct and analyze the model
as an AQFT.
They relate to multiple recent developments in AQFT and differ considerably from 
how \cite{GRW} construct the linear CS/WZW system as a factorization algebra.
Some of the key features are:
\begin{itemize}
\item[(1)] While the factorization algebra for the holomorphic
CS/WZW system is constructed in \cite{GRW} by linear BV quantization of a shifted Poisson structure (the antibracket),
our construction of an AQFT for the Lorentzian variant of this system uses canonical commutation
relation (CCR) quantization of a related \textit{unshifted} Poisson structure. We will determine
the latter by using the novel concept of Green's homotopies from \cite{BMS}, which are
homotopical generalizations of Green's operators, in combination with techniques
from $2$-dimensional chiral conformal Lorentzian geometry \cite{Grant-Stuart}.
In particular, we provide new examples of Green's homotopies for
gauge-theoretic bulk/boundary systems, which contributes novel types of examples to the literature on Green's 
operators in the presence of Lorentzian boundaries, see e.g.\ \cite{Dappiaggi,BDSboundary,Dappiaggi2}.

\item[(2)] Our AQFT for the CS/WZW bulk/boundary system is defined on a suitable category
of open subsets of the $3$-manifold $M$ and it satisfies a homotopical version of 
local constancy, or the time-slice axiom. Using recent strictification theorems
for the homotopy time-slice axiom \cite{BCStimeslice}, we will prove that this AQFT is
equivalent to an AQFT that is defined on a $2$-dimensional manifold $B$, playing a similar role
to a Cauchy surface. This dimensional reduction from $3$ to $2$ dimensions is crucial
to identify the degrees of freedom of the CS/WZW system near the boundary $\partial B$ 
with a chiral free boson, which in the context of AQFT is described on $1$-manifolds (i.e.\ light rays).
\end{itemize}

The outline of the remainder of this paper is as follows:
In Section \ref{sec:CSprelim}, we review the linear Chern-Simons complex
on an oriented $3$-manifold $M$ with boundary $\partial M\neq \emptyset$ from \cite{GRW} and introduce our Lorentz
geometric variant of the chiral WZW boundary condition, which we call the 
(anti-)self-dual boundary condition $\LLL^\pm$, see \eqref{eqn:boundarycondition}. In Section \ref{sec:greens_homotopy},
we define a suitable concept of chiral Green's homotopies for the Chern-Simons
complex and prove their existence under suitable niceness conditions, see Proposition \ref{prop:G_Theta:restricts_to_FL}.
The main tool we use to prove the existence of Green's homotopies
is an extension of the $\pm$-chiral flow $\Theta_\pm$ 
\cite{Grant-Stuart} on the boundary Lorentzian $2$-manifold $\partial M$ to a proper $\bbR$-action 
$\widehat{\Theta}_\pm: \bbR\times M\to M $ on the bulk manifold. One may interpret
this as an arbitrary extension of the concepts of $\pm$-chiral null curves from $\partial M$
into the bulk $M$. Equipped with these Green's homotopies, we construct
in Section \ref{sec:AQFTonM} an AQFT $\AAA_\pm$ for the CS/WZW bulk/boundary system
that assigns differential graded $\ast$-algebras (in short, dg-algebras)
to suitable open subsets $U\subseteq M$. We also prove that this AQFT
satisfies a chiral variant of the Einstein causality axiom and of the homotopy time-slice axiom,
see Proposition \ref{prop:propertiesPoissonfunctor} and Corollary \ref{cor:AQFTbulk}.
\sk

In Section \ref{sec:dimred}, we use the strictification theorems for the homotopy time-slice
axiom from \cite{BCStimeslice} in order to prove that our $3$-dimensional AQFT $\AAA_\pm$
for the CS/WZW bulk/boundary system on $M$ is equivalent to a $2$-dimensional AQFT $\BBB_\pm$
that is defined on the quotient space $B_\pm := M/\! \sim_{\pm}$ associated with
the proper $\bbR$-action $\widehat{\Theta}_\pm: \bbR\times M\to M $, see Corollary \ref{cor:dimreducedAQFT}. 
(Choosing a smooth section of the associated principal $\bbR$-bundle $\widehat{\pi}_\pm : M\to B_\pm$
allows one to think of $B_\pm$ as a kind of ``Cauchy surface'' in $M$.)
Using this phenomenon of dimensional reduction, we also prove that, under suitable topological assumptions on $M$,
the AQFT $\AAA_\pm$ for the CS/WZW bulk/boundary system is insensitive to the choice
of extension of the $\pm$-chiral flow $\Theta_\pm$ on $\partial M$ to a proper $\bbR$-action $\widehat{\Theta}_\pm$
on $M$, see Corollary \ref{cor:independence}. Hence, our AQFT construction is canonical in these cases.
\sk

The aim of Section \ref{sec:analysis} is to analyze the dimensionally reduced AQFT $\BBB_\pm$
that is defined on the quotient $2$-manifold $B_\pm := M/\! \sim_{\pm}$. 
Upon restriction to the interior $\mathrm{Int}B_\pm\subseteq B_\pm$, we find that
this AQFT specializes to the $2$-dimensional part of 
linear Chern-Simons theory, which is given by canonical quantization of the Chern-Simons 
phase space along the Atiyah-Bott Poisson structure \cite{AtiyahBott}.
Furthermore, restricting the AQFT $\BBB_\pm$ to a tubular neighborhood $\partial B_\pm \times[0,1)\subseteq B_\pm$
of the boundary, we obtain an AQFT that is equivalent to the chiral free boson on the $1$-manifold 
$\partial B_\pm$, see Corollary \ref{cor:chiralbosonQFT}. The interplay between
bulk and boundary degrees of freedom is realized in our approach through dg-algebra morphisms
$\BBB_\pm(V)\to\BBB_\pm(V^\prime)$ that are associated with subset inclusions
$V \subseteq V^\prime\subseteq B_\pm$ from interior regions $V\subseteq \mathrm{Int} B_\pm$
into regions $V^\prime\subseteq B_\pm$ that intersect the boundary, 
i.e.\ $V^\prime\cap\partial B_\pm\neq \emptyset$. This allows us to describe certain
Chern-Simons observables in the bulk in terms of equivalent chiral boson observables
on the boundary, which is illustrated in Example \ref{ex:holonomy} by studying
the boundary observable corresponding to the Chern-Simons observable that measures the holonomy
of flat connections. Appendix \ref{app:appendix} contains some technical details that
are required in the main part of the paper.

\paragraph{Notation and conventions for cochain complexes:} Let us fix a field $\bbK$
of characteristic $0$, e.g.\ the real numbers $\bbR$ or the complex numbers $\bbC$. 
A cochain complex $V$ consists a family of $\bbK$-vector spaces $\{V^i\}_{i\in\bbZ}$ and
a family of degree increasing linear maps $\{\dd : V^i\to V^{i+1}\}_{i\in\bbZ}$ (called the differential)
that squares to zero $\dd^2 =0$. A cochain map $f:V\to W$ between two cochain complexes
is a family of linear maps $\{f:V^i\to W^i\}_{i\in\bbZ}$ that commutes with the differentials
$f\circ \dd = \dd\circ f$. We denote by $\Ch_\bbK$ the category of cochain complexes and cochain maps.
\sk

Given any cochain complex $V$ and integer $p\in\bbZ$, we define the $p$-shifted 
cochain complex $V[p]$ by $V[p]^i := V^{i+p}$, for all $i\in\bbZ$, and the differential
$\dd_{[p]}:= (-1)^p\,\dd$.
\sk

To each cochain complex $V$ one can assign its cohomology, which is
the graded vector space defined by $\mathsf{H}^i(V):= \Ker(\dd:V^i\to V^{i+1})\big/\Imm(\dd:V^{i-1}\to V^i)$,
for all $i\in\bbZ$. The assignment of cohomology defines a functor $\mathsf{H}^\bullet$ from $\Ch_\bbK$
to the category of $\bbZ$-graded vector spaces. A morphism $f:V\to W$ in $\Ch_\bbK$ 
is called a quasi-isomorphism if it induces an isomorphism $\mathsf{H}^\bullet(f)
: \mathsf{H}^\bullet(V)\to \mathsf{H}^\bullet(W)$ in cohomology.
\sk

The category $\Ch_\bbK$ is closed symmetric monoidal. The monoidal product $V\otimes W$ 
of two cochain complexes is given by the family of vector spaces
\begin{subequations}
\begin{flalign}
(V\otimes W)^i \, :=\, \bigoplus_{j\in\bbZ} \left(V^j\otimes W^{i-j}\right)\quad,
\end{flalign}
for all $i\in\bbZ$, and the differential $\dd^\otimes$ that is 
defined by the graded Leibniz rule 
\begin{flalign}
\dd^\otimes(v\otimes w) \,:=\, (\dd v)\otimes w + (-1)^{\vert v\vert }\,v\otimes(\dd w)\quad,
\end{flalign}
\end{subequations}
for all homogeneous $v\in V$ and $w\in W$, where $\vert \cdot\vert\in\bbZ$ denotes the degree.
The monoidal unit is $\bbK$, regarded as a cochain complex concentrated in degree $0$ with trivial differential,
and the symmetric braiding $V\otimes W\to W\otimes V\, ,~v\otimes w\mapsto (-1)^{\vert v\vert\,\vert w\vert}\, w\otimes v$
is given by the Koszul sign rule.
The internal hom $[V,W]$ between two cochain complexes is given by the family of vector spaces
\begin{subequations}
\begin{flalign}
[V,W]^i\,:=\, \prod_{j\in\bbZ} \Hom_{\bbK}\big(V^j,W^{j+i}\big)\quad,
\end{flalign}
for all $i\in\bbZ$, where $\Hom_{\bbK}$ denotes the vector space of linear maps,
and the differential $\partial$ that is defined by the graded commutator 
\begin{flalign} \label{eqn:hom_differential}
\partial L\,:=\,\dd\circ L - (-1)^{\vert L\vert}\, L\circ \dd\quad,
\end{flalign}
\end{subequations}
for all homogeneous $L\in[V,W]$. 
\sk

Observe that a cochain map $f:V\to W$
is the same datum as a $0$-cocycle in the internal hom complex $[V,W]$, i.e.\ $f\in[V,W]^0$ such that
$\partial f =0$. Furthermore, a cochain homotopy between two cochain maps $f,g:V\to W$
is an element $h\in [V,W]^{-1}$ of degree $-1$ such that $\partial h = g-f$.
Higher cochain homotopies admit a similar interpretation in terms of the internal hom complex.


\section{\label{sec:CSprelim}Chern-Simons with (anti-)self-dual boundary condition}
In this section we describe the cochain complex that encodes linear Chern-Simons theory
on an oriented $3$-dimensional manifold $M$ and its $(-1)$-shifted symplectic structure.
The case where the boundary $\partial M =\emptyset$ is empty is quite standard
and details can be found e.g.\ in the book of Costello and Gwilliam \cite[Chapter 4.5]{CostelloGwilliam}.
The case where the boundary $\partial M \neq \emptyset$ is non-empty is more involved
as it requires the choice of a boundary condition to obtain a well-defined
$(-1)$-shifted symplectic structure. These aspects, including particular choices of boundary conditions,
have been studied systematically in the recent work \cite{GRW} of Gwilliam, Rabinovich and Williams,
which serves as our main reference for this section.

\paragraph{Empty boundary $\partial M=\emptyset$:}
Consider an oriented $3$-manifold $M$ without boundary $\partial M=\emptyset$.
The Chern-Simons complex is defined as the $1$-shifted de Rham complex
\begin{flalign}\label{eqn:fieldcomplex}
\FFF(M)\,:=\, \Omega^\bullet(M)[1] \,=\,\Big(
\xymatrix@C=2em{
\stackrel{(-1)}{\Omega^0(M)} \ar[r]^-{-\dd}
& \stackrel{(0)}{\Omega^1(M)}\ar[r]^-{-\dd} 
& \stackrel{(1)}{\Omega^2(M)} \ar[r]^-{-\dd}
& \stackrel{(2)}{\Omega^3(M)}
}
\Big)\quad,
\end{flalign}
where $\dd$ denotes the de Rham differential and round brackets indicate the cohomological degree.
It is important to note that the latter differs by $1$ from the de Rham degree, i.e.\
$\vert\alpha\vert_{\mathrm{dR}} = \vert \alpha\vert + 1$ for all homogeneous $\alpha\in\FFF(M)$.
The physical interpretation of the complex $\FFF(M)$ is as follows: 
An element $A\in \FFF(M)^0 = \Omega^1(M)$ in degree $0$ is a gauge field (principal $\bbR$-connection)
and an element $c\in \FFF(M)^{-1}=\Omega^0(M)$ in degree $-1$ is a ghost field 
(infinitesimal gauge transformation). The elements in positive degrees are the antifields
$A^{\ddagger}\in\FFF(M)^{1}=\Omega^2(M)$ and the antifields for ghosts $c^{\ddagger}\in \FFF(M)^{2}=\Omega^3(M)$.
As typical for the BV formalism, the differential encodes both the action of gauge symmetries
and the equation of motion, which in the case of linear 
Chern-Simons theory is the flatness condition  $\dd A =0$.
\sk

Using the wedge product of differential forms, we can define a cochain map
\begin{flalign}\label{eqn:wedgepairing}
(\,\cdot\,,\,\cdot\,)\,:\, \FFF(M)\otimes \FFF(M) ~\longrightarrow~ \Omega^\bullet(M)[2]~~,\quad
\alpha\otimes \beta ~\longmapsto~(\alpha,\beta)\,=\,  (-1)^{\vert \alpha\vert}\,\alpha\wedge \beta
\end{flalign}
to the $2$-shifted de Rham complex. Observe that this map is graded antisymmetric, i.e.\
\begin{flalign}
(\alpha,\beta) = - (-1)^{\vert \alpha\vert\,\vert \beta\vert}\,(\beta,\alpha)\quad,
\end{flalign}
for all homogeneous $\alpha,\beta\in\FFF(M)$. Since by hypothesis the boundary
$\partial M=\emptyset$ is empty, integration of compactly supported
forms defines a cochain map $\int_M : \Omega_\cc^\bullet(M)[2]\to \bbR[-1]$, which allows
us to define a $(-1)$-shifted symplectic structure 
\begin{flalign}\label{eqn:omega-1}
\omega_{(-1)} \,:\, \xymatrix@C=2.5em{
\FFF_\cc(M)\otimes\FFF_{\cc}(M) \ar[r]^-{(\,\cdot\,,\,\cdot\,)} ~&~ \Omega^\bullet_\cc(M)[2] \ar[r]^-{\int_M} ~&~
\bbR[-1]
}
\end{flalign}
on the subcomplex $\FFF_\cc(M)\subseteq \FFF(M)$ of compactly supported sections.

\paragraph{Non-empty boundary $\partial M\neq\emptyset$:}
Consider an oriented $3$-manifold $M$ with non-empty boundary $\partial M\neq \emptyset$.
As field complex we take again the $1$-shifted de Rham complex \eqref{eqn:fieldcomplex},
but now on a manifold with boundary. The 
$(-1)$-shifted symplectic structure  \eqref{eqn:omega-1} does \textit{not} directly generalize to
the present case because the map $\omega_{(-1)}$ fails to be a cochain map. 
Indeed, as a consequence of Stokes' theorem, one has that
\begin{flalign}\label{eqn:incompatibility}
-\partial \omega_{(-1)}(\alpha\otimes \beta) = 
\omega_{(-1)}\circ\dd^\otimes(\alpha\otimes \beta) = \int_{M}\dd (\alpha,\beta)=(-1)^{\vert \alpha\vert}\,\int_{\partial M} \iota^\ast(\alpha)\wedge\iota^\ast(\beta)\quad,
\end{flalign}
for all homogeneous $\alpha,\beta\in\FFF_\cc(M)$, where $\iota^\ast$ denotes the
pullback of differential forms along the boundary inclusion $\iota : \partial M \to M$.
\sk

One way to resolve this incompatibility between the differential and $\omega_{(-1)}$ 
is to impose a suitable boundary condition on the field complex $\FFF(M)$. A 
particularly interesting approach to boundary conditions, which 
draws inspiration from the intersection theory of Lagrangians
in derived algebraic geometry \cite{DAG}, has been recently
developed in \cite{GRW} and \cite{Rabinovich1,Rabinovich2}.
Without going into the technical details, which are nicely 
explained in these papers, the key idea of this approach
is to introduce a cochain complex $\FFF_\LLL(M)$ of 
\textit{boundary conditioned fields} by forming a fiber product
\begin{flalign}\label{eqn:conditionedfieldspullback}
\begin{gathered}
\xymatrix{
\ar@{-->}[d] \FFF_\LLL(M) \ar@{-->}[r] ~&~ \ar[d]^-{\iota^\ast}\FFF(M)\\
\LLL \ar[r]_-{\subseteq} ~&~ \FFF(\partial M)
}
\end{gathered}
\end{flalign}
in the category of cochain complexes.\footnote{The ordinary fiber
product \eqref{eqn:conditionedfieldspullback} provides a model for 
the homotopy fiber product whenever the right vertical map 
$\iota^\ast$ is a fibration of cochain complexes, which is 
the case in our example given by linear Chern-Simons theory.} 
Let us explain the ingredients 
of this diagram, as well as their interpretation, in the context of our example:
The boundary field complex $\FFF(\partial M)$ is the $1$-shifted 
de Rham complex of the boundary manifold $\partial M$, i.e.\
\begin{flalign}\label{eqn:boundarycomlpex}
\FFF(\partial M)\,:=\, \Omega^\bullet(\partial M)[1] \,=\,\Big(
\xymatrix@C=2em{
\stackrel{(-1)}{\Omega^0(\partial M)} \ar[r]^-{-\dd}
& \stackrel{(0)}{\Omega^1(\partial M)}\ar[r]^-{-\dd} 
& \stackrel{(1)}{\Omega^2(\partial M)} 
}
\Big)\quad,
\end{flalign}
and the cochain map $\iota^\ast$ is given by pullback of differential forms along $\iota :\partial M\to M$.
Using an analogous construction as in \eqref{eqn:omega-1}, together with the fact that 
$\partial M$ is $2$-dimensional, one obtains an unshifted symplectic structure
\begin{flalign}
\omega_{(0)}\,:\, \xymatrix@C=2.5em{
\FFF_\cc(\partial M)\otimes\FFF_{\cc}(\partial M) \ar[r]^-{(\,\cdot\,,\,\cdot\,)} ~&~ 
\Omega^\bullet_\cc(\partial M)[2] \ar[r]^-{\int_{\partial M}} ~&~ \bbR
}
\end{flalign}
on the subcomplex $\FFF_\cc(\partial M)\subseteq \FFF(\partial M)$ of compactly supported sections.
This allows us to rewrite Eqn.~\eqref{eqn:incompatibility} in the following more conceptual form
\begin{flalign}\label{eqn:incompatibility2}
-\partial\omega_{(-1)} = \omega_{(0)}\circ (\iota^\ast\otimes\iota^\ast)\quad,
\end{flalign}
which means that $\omega_{(-1)}$ defines a derived isotropic structure for the cochain map $\iota^\ast$
in \eqref{eqn:conditionedfieldspullback}. As a consequence of the non-degeneracy condition proven in
\cite[Lemma 2.20]{Rabinovich1}, $\omega_{(-1)}$ is further (a field theoretic analogue of)
a derived Lagrangian structure on $\iota^\ast$. 
\sk

The subcomplex $\LLL\subseteq \FFF(\partial M)$ 
in \eqref{eqn:conditionedfieldspullback}, which we call the \textit{boundary condition},
is the choice of a local and strict Lagrangian in $\big(\FFF(\partial M),\omega_{(0)}\big)$.
Local means that $\LLL = \Gamma^\infty(L)$ is given by the sections of a subbundle $L\to\partial M$ 
of the boundary field bundle $\wedge^\bullet T^\ast\partial M[1] \to \partial M$, and
strict Lagrangian means that $\LLL\subseteq \FFF(\partial M)$ is strictly isotropic,
i.e.\ the restriction
\begin{flalign}\label{eqn:LLLisotropic}
\omega_{(0)}\big\vert_{\LLL_\cc\otimes \LLL_\cc}^{~} \,=\, 0
\end{flalign}
of the unshifted symplectic structure to the subcomplex $\LLL_\cc\subseteq \LLL$ of 
compactly supported sections vanishes, and maximal in the sense that the total rank of 
$L$ is half that of $\wedge^\bullet T^\ast\partial M[1] $.
\sk

An explicit model for the cochain complex $\FFF_{\LLL}(M)$ defined by the fiber product in
\eqref{eqn:conditionedfieldspullback} is given by the 
subcomplex $\FFF_{\LLL}(M)\subseteq \FFF(M)$ of the field complex consisting of
all fields $\alpha\in \FFF(M)$ that satisfy the boundary condition $\iota^\ast(\alpha)\in\LLL$.
We denote the individual components of this subcomplex by
\begin{flalign}\label{eqn:fieldcomplex:boundarycondition}
\FFF_{\LLL}(M) \,=\, \Omega_\LLL^\bullet(M)[1] \,=\,\Big(
\xymatrix@C=2em{
\stackrel{(-1)}{\Omega_\LLL^0(M)} \ar[r]^-{-\dd}
& \stackrel{(0)}{\Omega_\LLL^1(M)}\ar[r]^-{-\dd} 
& \stackrel{(1)}{\Omega_\LLL^2(M)} \ar[r]^-{-\dd}
& \stackrel{(2)}{\Omega_\LLL^3(M)}
}
\Big)\quad.
\end{flalign}
Due to the strict isotropy condition \eqref{eqn:LLLisotropic} for $\LLL$, one 
immediately observes that the boundary terms in \eqref{eqn:incompatibility}
and \eqref{eqn:incompatibility2} vanish for fields that satisfy the boundary condition, 
hence $\omega_{(-1)}$ restricts to a $(-1)$-shifted symplectic structure
\begin{flalign}
\omega_{(-1)}\,:\,\FFF_{\LLL,\cc}(M)\otimes \FFF_{\LLL,\cc}(M)~\longrightarrow~\bbR[-1]
\end{flalign}
on the subcomplex $\FFF_{\LLL,\cc}(M)\subseteq \FFF_\LLL(M)$ of compactly supported sections.
\sk

Note that, when the boundary $\partial M =\emptyset$ is empty, 
the boundary field complex $\FFF(\partial M) =0$ is trivial and so is any boundary condition $\LLL=0$.
Hence, the cochain complex $\FFF_\LLL(M)$ and its $(-1)$-shifted symplectic structure 
$\omega_{(-1)}$ specialize to the ones we have described in the previous paragraph.

\paragraph{Linear observables and $(-1)$-shifted Poisson structure:} We describe
the linear observables for the complex $\FFF_\LLL(M)$ of boundary conditioned fields
by the $1$-shifted cochain complex
\begin{flalign}\label{eqn:linobs}
\FFF_{\LLL,\cc}(M)[1]\,=\,\Omega^\bullet_{\LLL,\cc}(M)[2] \,=\, 
\Big(
\xymatrix@C=2em{
\stackrel{(-2)}{\Omega_{\LLL,\cc}^0(M)} \ar[r]^-{\dd}
& \stackrel{(-1)}{\Omega_{\LLL,\cc}^1(M)}\ar[r]^-{\dd} 
& \stackrel{(0)}{\Omega_{\LLL,\cc}^2(M)} \ar[r]^-{\dd}
& \stackrel{(1)}{\Omega_{\LLL,\cc}^3(M)}
}
\Big)
\end{flalign}
of compactly supported sections. The evaluation of linear observables
on fields is given by the cochain map
\begin{subequations}\label{eqn:ev}
\begin{flalign}
\ev\,:\, \xymatrix@C=2.5em{
\FFF_{\LLL,\cc}(M)[1]\otimes \FFF_\LLL(M)\ar[r]^-{(\,\cdot\,,\,\cdot\,)}~&~\Omega_\cc^\bullet(M)[3]\ar[r]^-{\int_M}~&~\bbR
}
\end{flalign}
that is obtained from \eqref{eqn:wedgepairing} and the integration map. Explicitly,
\begin{flalign}
\ev(\varphi\otimes \alpha) \,= \,(-1)^{\vert \varphi\vert +1} \, \int_M \varphi\wedge\alpha\quad,
\end{flalign}
\end{subequations}
for all homogeneous $\varphi\in \FFF_{\LLL,\cc}(M)[1]$ and $\alpha\in  \FFF_\LLL(M)$.
The complex of linear observables can be endowed with the
$(-1)$-shifted Poisson structure
\begin{flalign}\label{eqn:shiftedPoisson}
\tau_{(-1)} \,:\, 
\xymatrix@C=2.5em{
\FFF_{\LLL,\cc}(M)[1]\otimes \FFF_{\LLL,\cc}(M)[1] \ar[r]^-{\cong}~&~\FFF_{\LLL,\cc}(M)[2]\otimes \FFF_{\LLL,\cc}(M)\ar[r]^-{\omega_{(-1)}}~&~\bbR[1]\quad,
}
\end{flalign}
where the first cochain isomorphism is constructed from the symmetric braiding on the
category of cochain complexes, i.e.\ it involves Koszul signs.
These Koszul signs cancel the minus signs entering $\omega_{(-1)}$ through \eqref{eqn:wedgepairing},
such that
\begin{flalign}
\tau_{(-1)} (\varphi\otimes\psi) = \int_M\varphi\wedge\psi
\end{flalign}
is simply given by the forming the wedge product followed by integration over $M$,
for all linear observables  $\varphi,\psi\in \FFF_{\LLL,\cc}(M)[1]$.

\paragraph{(Anti-)self-dual boundary condition:}
Of particular interest to us are Lorentzian geometric variants of
the so-called chiral and anti-chiral WZW boundary conditions,
which were introduced in \cite[Example 2.3]{GRW} through the choice of a complex structure on $\partial M$ 
and the corresponding Dolbeault decomposition of $\Omega^\bullet(\partial M)$.
To describe a Lorentzian analogue of these boundary conditions, we assume that 
$\partial M$ comes endowed with (the conformal class of) a Lorentzian 
metric $g$ and a time orientation. As orientation on $\partial M$ 
we take the one induced from the orientation of the bulk $M$.
From these data we can define a Hodge operator $\ast : \Omega^k(\partial M)
\to \Omega^{2-k}(\partial M)$, which due to the Lorentzian signature of the metric 
squares to the identity $\ast^2 = \id$ on $1$-forms. This allows us to decompose 
the vector space of $1$-forms
\begin{flalign}
\Omega^1(\partial M) \,= \, \Omega^{+}(\partial M) \oplus \Omega^{-}(\partial M)
\end{flalign}
into the self-dual $\Omega^+(\partial M)$ and anti-self-dual 
$ \Omega^-(\partial M)$ 
forms, i.e.\ $\alpha\in \Omega^{\pm}(\partial M)\subseteq \Omega^1(\partial M)$ 
if and only if $\ast\alpha = \pm\alpha$. Let us note that, since the Hodge operator is
$C^\infty(\partial M)$-linear, this decomposition is also 
a direct sum decomposition of $C^\infty(\partial M)$-modules. In particular, 
$\Omega^\pm(\partial M)\subseteq \Omega^1(\partial M)$ are submodules.
\begin{ex}\label{ex:Minkowski1}
In order to illustrate these concepts, let us consider the $3$-manifold
$M = \bbR^2\times \bbR^{\geq 0}$ with orientation $\mathfrak{o} = [-\dd t\wedge \dd x \wedge \dd r]$
determined by a choice of global coordinates $(t,x,r)$. Let us further endow the boundary 
$\partial M = \bbR^2$ with the standard Minkowski metric
$g = -\dd t^2 + \dd x^2$ and the time-orientation
$\mathfrak{t} = [\partial_t]$. The induced orientation of the boundary is
$\mathfrak{o}_{\partial M} = [\dd t\wedge \dd x]$.
Taking the standard basis $(\dd t, \dd x)$ of $\Omega^1(\partial M)$, one has that
$\ast \dd t = - \dd x$ and $\ast \dd x = - \dd t$.
Introducing the standard light cone coordinates $x^\pm := t\pm x$,
the Hodge operator acts on $\dd x^\pm\in \Omega^1(\partial M)$
according to $\ast(\dd x^\pm) = \mp \dd x^\pm$, i.e.\ $\dd x^-$ defines a module basis for
the self-dual forms $\Omega^+(\partial M)$ and $\dd x^+$ defines one for the anti-self-dual forms $\Omega^-(\partial M)$.
Note the flip of sign. 
\end{ex}

We define the \textit{self-dual boundary condition} as the subcomplex
\begin{subequations}\label{eqn:boundarycondition}
\begin{flalign}
\LLL^{+} \,:=\, \Big(
\xymatrix@C=2em{
\stackrel{(-1)}{0} \ar[r]^-{0}
& \stackrel{(0)}{\Omega^{+}(\partial M)}\ar[r]^-{-\dd} 
& \stackrel{(1)}{\Omega^2(\partial M)} 
}
\Big)\, \subseteq \, \FFF(\partial M)
\end{flalign}
of the boundary field complex \eqref{eqn:boundarycomlpex}.
The \textit{anti-self-dual boundary condition} $\mathfrak{L}^-$ is similarly 
defined with $\Omega^-(\partial M)$ in place of $\Omega^+(\partial M)$, i.e.\
\begin{flalign}
\, \quad
\LLL^{-} \,:=\, \Big(
\xymatrix@C=2em{
\stackrel{(-1)}{0} \ar[r]^-{0}
& \stackrel{(0)}{\Omega^{-}(\partial M)}\ar[r]^-{-\dd} 
& \stackrel{(1)}{\Omega^2(\partial M)} 
}
\Big)\, \subseteq \, \FFF(\partial M)
\quad.
\end{flalign}
\end{subequations}
It is easy to check that both $\mathfrak{L}^+$ and $\mathfrak{L}^-$ are local and strict Lagrangians
in $\big(\FFF(\partial M),\omega_{(0)}\big)$.
\sk

Throughout the whole paper, we impose either $\mathfrak{L}^+$ or $\mathfrak{L}^-$
as boundary condition on the Chern-Simons theory on $M$.\footnote{In 
cases where $\partial M$ has multiple components,
one may choose a more general mixed-type boundary condition
by using self-dual forms on some components of $\partial M$
and anti-self-dual forms on the others.
While much of the discussion below may be adapted to this case,
we give it no further consideration here.}
For a uniform notation,
we denote the chosen boundary condition as $\mathfrak{L}^\pm$.
The corresponding boundary conditioned fields introduced in \eqref{eqn:conditionedfieldspullback}
then satisfy the following boundary conditions:
The pullback to the boundary $\iota^\ast (c) =  0$ of the ghost field
vanishes and the pullback to the boundary
$\iota^\ast(A) \in \Omega^{\pm}(\partial M)$ of the gauge field
is an (anti-)self-dual $1$-form.
Note that the antifields $A^\ddagger$ and $c^\ddagger$
are not restricted by the (anti-)self-dual boundary condition.


\section{\label{sec:greens_homotopy}Construction of Green's homotopies}
The data presented in Section \ref{sec:CSprelim} is sufficient
to quantize linear Chern-Simons theory with the chosen (anti-)self-dual
boundary condition in terms of a factorization algebra. The relevant
construction is given by linear BV quantization of the complex
of linear observables \eqref{eqn:linobs} along its $(-1)$-shifted Poisson structure
\eqref{eqn:shiftedPoisson}, see e.g.\ \cite[Chapter 4]{CostelloGwilliam}
or \cite[Section 3]{GRW} for the details.
\sk

In the present paper we shall take a different approach and study the 
quantization of this theory in terms of an AQFT. This requires additional 
structure, given by so-called \textit{Green's homotopies} \cite{BMS},
which generalize the concept of retarded/advanced Green's operators
to a homotopical context. Using these Green's homotopies, one can
determine as in \cite{BMS} an \textit{unshifted} Poisson structure
on the complex of linear observables \eqref{eqn:linobs}, which then can be
quantized by constructing
canonical commutation relation (CCR) dg-algebras  as in \cite{LinearYM}.
\begin{ex}
Green's homotopies provide a natural homotopical generalization of the usual retarded/advanced 
Green's operators for normally hyperbolic partial differential equations.
To illustrate the main idea, let us briefly sketch the concept of Green's homotopies
in the simplest example given by Klein-Gordon theory on a globally hyperbolic Lorentzian manifold $N$. 
We refer the reader to \cite{BMS} and also \cite[Section 4.1]{LinearYM} for more details.
The field complex of the Klein-Gordon field is given by
\begin{flalign}\label{eqn:KGcomplex}
	\mathfrak{F}^\mathrm{KG}(N)
	\, := \,
	\Big(
	\xymatrix{
	\stackrel{(0)}{\Omega^0(N)}
	\ar[r]^-{\square + m^2}
	&
	\stackrel{(1)}{\Omega^0(N)}
	\Big)}
	\quad,
\end{flalign}
where $\square$ denotes the d'Alembertian and $m^2\geq 0$ is a mass term.
Elements $\Phi\in\Omega^0(N)$ in degree $0$ are interpreted
as fields and elements $\Phi^\ddagger\in\Omega^0(N)$ in degree $1$ as antifields.
Note that the cohomology of this complex is concentrated in degree $0$ and it is given by
the usual solution space of the Klein-Gordon equation
$\Sol(N) = \{\Phi\in \Omega^0(N)\,\vert\,(\square+m^2)\Phi=0\}$.
\sk

A retarded/advanced Green's operator for Klein-Gordon theory is defined as a
linear map $G^\pm : \Omega^0_\mathrm{c}(N) \to \Omega^0(N)$ that satisfies, 
for all compactly supported $\varphi \in  \Omega^0_\mathrm{c} (N)$,
\begin{enumerate}[(i)]
	\item $G^\pm (\square + m^2) \varphi = \varphi$ and $(\square + m^2) G^\pm \varphi = \varphi$,

	\item $\operatorname{supp} (G^\pm \varphi) \subseteq J^\pm_N (\operatorname{supp} \varphi)$, where $J^\pm_N(S) \subseteq N$ denotes the causal future/past of $S \subseteq N$.
\end{enumerate}
The first property can be easily described in the language of cochain complexes:
The retarded/advanced Green's operator determines via
\begin{flalign}
\begin{gathered}
	\xymatrix@C=3em@R=0.3em{
	\text{\footnotesize (-1)} & \text{\footnotesize (0)} & \text{\footnotesize (1)} & \text{\footnotesize (2)} \\
	\ar[dddd] {0} \ar[r] 
	&
	\ar[dddd]_-{j}{\Omega_\mathrm{c}^0(N)}
	\ar[r]^-{\square + m^2} 
	\ar@{-->}[ldddd]
	&
	\ar[dddd]_-{j}{\Omega_\mathrm{c}^0(N)}
	\ar[r] 
	\ar@{-->}[ldddd]|{G^\pm}
	&
	\ar[dddd] {0}
	\ar@{-->}[ldddd]
	\\ \\ \\ \\
	0 \ar[r]
	&
	\Omega^0(N) \ar[r]_-{\square + m^2}
	&
	\Omega^0(N) \ar[r]
	&
	0
	}
\end{gathered}
\end{flalign}
a $(-1)$-cochain 
$G^\pm\in \big[\mathfrak{F}_\cc^\mathrm{KG}(N),\mathfrak{F}^\mathrm{KG}(N)\big]^{-1}$
which defines a cochain homotopy $\partial G^\pm = j$
from the zero map to the cochain map $j : \mathfrak{F}^\mathrm{KG}_\mathrm{c}(N) \hookrightarrow \mathfrak{F}^\mathrm{KG}(N)$ 
that includes compactly supported $0$-forms into all $0$-forms.
The second property of Green's operators states that this homotopy must have 
suitable support properties, which can be formalized in a homotopically meaningful way as in \cite{BMS}.
Taking the difference $G:=G^{+}-G^{-}\in \big[\mathfrak{F}_\cc^\mathrm{KG}(N),\mathfrak{F}^\mathrm{KG}(N)\big]^{-1}$
of the retarded and the advanced Green's homotopy defines a cochain map $G : \mathfrak{F}_\cc^\mathrm{KG}(N)[1]\to
\mathfrak{F}^\mathrm{KG}(N)$, i.e.\ $\partial G = 0$, from which one constructs an unshifted
Poisson structure $\tau_{(0)} : \mathfrak{F}_\cc^\mathrm{KG}(N)[1]\otimes \mathfrak{F}_\cc^\mathrm{KG}(N)[1]\to\bbR$
via integration $\tau_{(0)}(\varphi\otimes \psi) :=\int_N \varphi\,G(\psi)\,\vol_N $, where $\vol_N$ denotes
the volume form.
\sk

The advantage of Green's homotopies is that they exist also in gauge-theoretic examples, 
which are described by complexes that encode also ghosts and their antifields, 
and hence are longer than \eqref{eqn:KGcomplex}, while ordinary Green's operators 
can not exist in a gauge theory (before gauge fixing)
due to degeneracy of the equation of motion.
\end{ex}

Our strategy in this section is as follows:
In Subsection \ref{subsec:prep_for_GH} we use relevant chiral geometry (in the sense of \cite{Grant-Stuart})
of the boundary Lorentzian manifold $\partial M$ to define
a notion of Green's homotopies for the (anti-)self-dual boundary condition
$\mathfrak{L}^\pm$ in \eqref{eqn:boundarycondition}.
A key feature of this chiral geometry is the so-called \textit{chiral flows}. 
In Subsection \ref{subsec:Greens_homotopies_from_flow} we abstract to Green's homotopies 
for the de Rham complex on a manifold that is endowed with a smooth $\bbR$-action (i.e.\ a flow),
and give an explicit construction of such Green's homotopies under suitable niceness conditions on the $\bbR$-action.
Finally, in Subsection \ref{subsec:GH_on_boundary_and_bulk} we specialize this construction to produce compatible
Green's homotopies for the boundary condition $\mathfrak{L}^\pm$  
and the complex $\mathfrak{F}_{\mathfrak{L}^\pm}(M)$ of boundary conditioned fields,
under the hypothesis that the bulk $M$ may be equipped 
with suitable extensions of a chiral flow on $\partial M$.

\subsection{\label{subsec:prep_for_GH}Chiral geometry of the boundary}
To motivate our notion of Green's homotopies for Chern-Simons theory subject to an
(anti-)self-dual boundary condition, we first consider the analogous notion of Green's 
homotopies for the boundary condition complex $\mathfrak{L}^\pm$ itself.
In so doing, we will encounter relevant chiral features of the $2$-dimensional 
boundary spacetime $\partial M$, which will inform our construction in the bulk $M$.
\sk

In this subsection we will work on the $2$-dimensional boundary manifold $\partial M$,
which by our hypotheses is endowed with a Lorentzian metric $g$, an orientation $\mathfrak{o}_{\partial M}$
and a time-orientation $\mathfrak{t}$.
In the following definition we represent the latter two data by
a non-vanishing $2$-form $\omega\in \Omega^2(\partial M)$ and a non-vanishing
time-like vector field $\tau\in\Gamma^\infty(T\partial M)$.
\begin{defi}\label{def:nullpointing}
\begin{itemize}
\item[(a)] A non-zero null vector $0\neq v\in T_p\partial M$ at a point $p\in\partial M$, i.e.\ $g(v,v)=0$, is called
\textit{$\nearrow$-pointing} if $g(\tau,v)<0$ and $\omega(\tau,v)>0$,
\textit{$\nwarrow$-pointing} if $g(\tau,v)<0$ and $\omega(\tau,v)<0$,
\textit{$\swarrow$-pointing} if $g(\tau,v)>0$ and $\omega(\tau,v)<0$, and
\textit{$\searrow$-pointing} if $g(\tau,v)>0$ and $\omega(\tau,v)>0$.

\item[(b)] A null curve $\gamma : [0,1]\to \partial M$ is called \textit{$\nearrow$-pointing}
if all its tangent vectors are $\nearrow$-pointing. The 
other three cases $\nwarrow$, $\swarrow$ and $\searrow$ are defined similarly.
\end{itemize}
\end{defi}

Evidently, any $\nearrow$-pointing curve $\gamma$ produces a $\swarrow$-pointing curve 
$s \mapsto \gamma(1-s)$ by reversing its parameterization, and vice versa.
$\nwarrow$- and $\searrow$-pointing curves are similarly related.
\begin{ex}\label{ex:Minkowski2}
The terminology in Definition \ref{def:nullpointing} is motivated
by the graphical representation of such tangent vectors in the standard Minkowski spacetime 
(see Example \ref{ex:Minkowski1}):
\begin{flalign}
\begin{gathered}
\begin{tikzpicture}[scale=0.4]
\draw[thick,->] (-4,0) -- (4,0);
\draw[thick,->] (0,-4) -- (0,4);
\draw[thick,->] (-3,-3) -- (3,3);
\draw[thick,->] (3,-3) -- (-3,3);
\draw (0,4.5) node{$t$};
\draw (4.5,0) node{$x$};
\draw (3.75,3.5) node{$x^+$};
\draw (-3.25,3.5) node{$x^-$};
\end{tikzpicture}
\end{gathered}
\end{flalign}
Indeed, by a quick calculation one checks that
the tangent vector $\partial_+ := \frac{\partial}{\partial x^+}$ is $\nearrow$-pointing and
the tangent vector $\partial_- := \frac{\partial}{\partial x^-}$ is $\nwarrow$-pointing. 
The additive inverses $-\partial_+$ and $-\partial_-$ are, respectively, $\swarrow$-pointing
and $\searrow$-pointing.
\end{ex}

\begin{rem}
In \cite[Section 4]{Grant-Stuart},
a null curve $\gamma:[0,1]\to\partial M$ that is either 
$\nearrow$-pointing or $\swarrow$-pointing is called \textit{right-chiral},
and one that is either $\nwarrow$-pointing or $\searrow$-pointing is called \textit{left-chiral}.
In the present paper we prefer to use the finer distinction from Definition \ref{def:nullpointing}.
\end{rem}

Our concept of pointings for tangent vectors is related to
(anti-)self-duality of $1$-forms.
\begin{lem}\label{lem:sd/asdannihilation}
A $1$-form $\alpha \in \Omega^1(\partial M)$ is self-dual if and only if, at every point $p\in\partial M$,
the contraction $\alpha(v) = 0 $ vanishes for all $\nearrow$-pointing and all 
$\swarrow$-pointing null vectors $v\in T_p\partial M$. Similarly,
a $1$-form $\alpha \in \Omega^1(\partial M)$ is anti-self-dual
if and only if, at every point $p\in\partial M$, the contraction $\alpha(v) = 0 $
vanishes for all $\nwarrow$-pointing and all $\searrow$-pointing null vectors $v\in T_p\partial M$. 
\end{lem}
\begin{proof}
Since every $2$-dimensional Lorentzian manifold is locally 
conformally equivalent to the Minkowski spacetime,
and since (anti-)self-duality and the pointings from Definition \ref{def:nullpointing}
depend only on the conformal structure, it is sufficient to prove the claim for
the Minkowski spacetime. The latter follows immediately from our explicit descriptions 
in Examples \ref{ex:Minkowski1} and \ref{ex:Minkowski2}. 
\end{proof}

Making use of the pointings from Definition \ref{def:nullpointing},
we can introduce for each point $p\in\partial M$ the following subset
\begin{subequations}\label{eqn:J_ne}
\begin{flalign} \label{eqn:J_ne:point}
J^{\ne}(p)\,:=\, \left\{q\in\partial M
\,\middle\vert\,
\begin{array}{c}
	q=p \text{ or } \exists \text{ $\nearrow$-pointing curve } \\
	\gamma:[0,1]\to \partial M \text{ s.t.\ }\gamma(0)=p\text{ and } \gamma(1)=q
\end{array}
\right\}\,\subseteq\,\partial M\quad,
\end{flalign}
and similar subsets for the other three cases $\nwarrow$, $\swarrow$ and $\searrow$.
These subsets characterize the four null components of the light cone at $p\in\partial M$.
Given any subset $S\subseteq \partial M$, we further define
\begin{flalign} \label{eqn:J_ne:set}
J^{\ne}(S)\,:=\, \bigcup_{p\in S} J^{\ne}(p)\,\subseteq \, \partial M\quad,
\end{flalign}
\end{subequations}
and similar in the other three cases.
\begin{rem} \label{rem:ne/sw_relation}
Since reversing the parameterization changes a $\nearrow$-pointing curve into a 
$\swarrow$-pointing curve, it follows that $p \in J^\ne(q)$ if and only if $q \in J^\sw(p)$.
In fact, $J^\ne(p)$ and $J^\sw(p)$ arise respectively as the sets of successors and 
predecessors of $p$ under the binary relation $\leq_+$ on $\partial M$ defined by
$p \leq_+ q$ if either $p = q$ or there exists a $\nearrow$-pointing curve in $\partial M$ from $p$ to $q$.
Similarly, $J^\nw(p)$ and $J^\se(p)$ arise respectively as the sets of successors and predecessors
of $p$ under the binary relation $\leq_-$ on $\partial M$ defined by 
$p \leq_- q$ if either $p = q$ or there exists a $\nwarrow$-pointing curve in $\partial M$ from $p$ to $q$.
\end{rem}

In the next definition we introduce the variants of Green's homotopies 
that are relevant for the present paper. See also
Remark \ref{rem:Greenshomotopies} below for some additional comments 
on the relationship to the framework of Green's homotopies
developed in \cite{BMS}.
\begin{defi}\label{def:Greenshomotopies_on_partial_M}
A \textit{$\nearrow$-pointed Green's homotopy} for the boundary condition
$\LLL^\pm$ in \eqref{eqn:boundarycondition} is a $(-1)$-cochain
$G^{\ne} \in [\LLL_\cc^\pm,\LLL^\pm]^{-1}$ in the internal hom complex
from the complex of compactly supported sections $\LLL_\cc^\pm$ to $\LLL^\pm$
that satisfies the following properties:
\begin{enumerate}[(i)]
\item \label{def:Greenshomotopies_on_partial_M:eom}
$\partial G^{\ne} = j$, where $j: \LLL_\cc^\pm \to \LLL^\pm$
is the canonical inclusion cochain map from compactly supported sections to all sections;

\item \label{def:Greenshomotopies_on_partial_M:support}
$\supp\big(G^{\ne}\alpha\big) \subseteq J^{\ne}\big(\supp\, \alpha\big)$, for all
$\alpha\in \LLL^\pm_\cc$.
\end{enumerate}
The definition of Green's homotopies for the other three pointings $\nwarrow$, $\swarrow$ and $\searrow$ is similar.
\end{defi}

\begin{rem}\label{rem:Greenshomotopies}
The following remarks are in order:
\begin{enumerate}
\item There are four types of Green's homotopies referring to the four different components
of the light cone in $2$-dimensional conformal Lorentzian geometry. These refine the two types
of Green's homotopies (retarded and advanced) in standard $m$-dimensional Lorentzian geometry.

\item The definition of Green's homotopies in \cite[Definition 3.5]{BMS} is more general 
as it allows for pseudo-naturality under inclusions $K\subseteq K^\prime$ of compact subsets of 
$\partial M$. From a conceptual perspective, this generality is needed to prove a 
uniqueness theorem for Green's homotopies stating that the \textit{space of Green's homotopies} 
is either empty or contractible, see \cite[Proposition 3.9]{BMS}. This level of generality
is not needed in the present paper since all of our Green's homotopies will be
of the strict type as defined in Definition \ref{def:Greenshomotopies_on_partial_M}.
\end{enumerate}
Later in Proposition \ref{prop:G_Theta:restricts_to_L} we will show that, 
under suitable chiral niceness properties of $\partial M$, 
the self-dual boundary condition $\LLL^+$ admits $\nearrow$- and $\swarrow$-pointed Green's homotopies
and the anti-self-dual boundary condition $\LLL^-$ admits $\nwarrow$- and $\searrow$-pointed ones.
\end{rem}

As a preparation for our construction of Green's homotopies,
let us recall from \cite[Proposition 4.14]{Grant-Stuart} that on each oriented
and time-oriented Lorentzian $2$-manifold $\partial M$ there exists a so-called
\textit{chiral frame} $(n_-,n_+)$ for the tangent bundle $T\partial M$, 
i.e.\ $n_-\in \Gamma^\infty(T\partial M)$ is an everywhere $\nwarrow$-pointing
null vector field and $n_+\in  \Gamma^\infty(T\partial M)$ is an everywhere $\nearrow$-pointing
vector field. The choice of chiral frame is canonical up to rescalings 
$n_-\sim f_-\,n_-$ and $n_+\sim f_+\,n_+$ by positive smooth functions $f_\mp \in C^\infty(\partial M,\bbR^{>0})$.
Furthermore, by \cite[Lemma 4.19]{Grant-Stuart}, one may assume without loss of generality
that both $n_-$ and $n_+$ are complete vector fields. This allows us to define
the \textit{$+$- and $-$-chiral flows}
\begin{flalign}
\Theta_{+}, \Theta_- \,:\, \bbR\times \partial M~\longrightarrow~\partial M
\end{flalign}
as the  global flows of the complete vector fields $n_+$ and $n_-$ respectively.
\begin{ex}\label{ex:Minkowski:flows}
Coming back to the Minkowski spacetime from Examples \ref{ex:Minkowski1} and \ref{ex:Minkowski2},
a complete chiral frame is given by $(\partial_-, \partial_+)$. The associated chiral flows
read in light cone coordinates as $\Theta_+(s,(x^-,x^+)) = (x^-,x^+ +s)$
and $\Theta_{-}(s,(x^-x^+)) = (x^-+s,x^+)$, i.e.\ they are simply translations along the light cone coordinates.
\end{ex}

Chiral flows provide a useful technical tool to handle
the variously-pointed null curves of Definition \ref{def:nullpointing}.
By \cite[Lemma 4.21]{Grant-Stuart},
any $\nearrow$- or $\swarrow$-pointing null curve in $\partial M$ is,
up to reparameterization, an integral curve associated to the $+$-chiral flow $\Theta_+$.
Similarly, the $\nwarrow$- and $\searrow$-pointing null curves are described by the $-$-chiral flow $\Theta_-$.
\begin{rem}\label{rem:canonical_flow}
There exists no canonical choice of the chiral flows $\Theta_+$ and $\Theta_-$ on $\partial M$,
owing to the freedom to rescale $n_+$ and $n_-$ by positive smooth functions. 
However, the orbits of any two $+$-chiral flows $\Theta_+$ and $\Theta_+^\prime$ 
coincide, being the images of inextendable $\nearrow$-pointing null curves in $\partial M$.
Thus, the quotient $\pi_+ : \partial M \to \partial M / \! \sim_+$
is canonical on a given $2$-dimensional spacetime $\partial M$,
for $\sim_+$ the orbit relation of any $+$-chiral flow.
Similarly, the quotient $\pi_- : \partial M \to \partial M /\! \sim_-$ by any $-$-chiral flow is canonical.
\end{rem}

\begin{rem}\label{rem:J_from_flow}
One may describe the sets $J^\ne(p)$ and $J^\sw(p)$ from
\eqref{eqn:J_ne:point} in terms of the $+$-chiral flow only. Indeed, by
\cite[Proposition 4.22]{Grant-Stuart}, a point $q\in\partial M$ lies in $J^\ne(p)$
(or equivalently, $p\in\partial M$ lies in $J^\sw(q)$) if and only if $q = \Theta_+(s,p)$
for some $s \geq 0$. The sets $J^\nw(p)$ and $J^\se(p)$ may be characterized
similarly using the $-$-chiral flow.
\end{rem}

Just as we denote by $\mathfrak{L}^\pm$ our choice of self-dual 
or anti-self-dual boundary condition, we henceforth denote by 
$\Theta_\pm$ a corresponding choice of $+$- or $-$-chiral flow 
on $\partial M$ ($+$-chiral if the chosen boundary condition is self-dual, 
$-$-chiral if the boundary condition is anti-self-dual).

\subsection{\label{subsec:Greens_homotopies_from_flow}Green's homotopies associated with a proper \texorpdfstring{$\bbR$}{R}-action}
Our construction of Green's homotopies for $\mathfrak{L}^\pm$ 
will use only the chiral flow $\Theta_\pm$ on $\partial M$,
without directly referring to any of the other available data, e.g.\
the Lorentzian metric, orientation or time-orientation. 
To simplify notation during this construction,
we temporarily abstract to the $k$-shifted de Rham complex $\Omega^\bullet(N)[k]$ 
of an arbitrary smooth manifold $N$ (with or without boundary)
that is endowed with a smooth $\bbR$-action (i.e.\ flow) $\Theta : \bbR \times N \to N$.
Our general construction can be specialized in the case of $(N, \Theta) = (\partial M, \Theta_\pm)$
to provide appropriately-pointed Green's homotopies for $\mathfrak{L}^\pm \subseteq \Omega^\bullet(\partial M)[1]$.
This abstraction will also be of use later in Proposition \ref{prop:G_Theta:restricts_to_FL},
when we will take $N = M$ to build appropriate Green's homotopies 
for the complex $\mathfrak{F}_{\mathfrak{L}^\pm}(M) \subseteq \Omega^\bullet(M)[1]$ 
of boundary conditioned Chern-Simons fields.
\sk

Given any manifold $N$ that is endowed with an $\bbR$-action 
$\Theta:\bbR\times N\to N$, we can define analogues of the subsets $J^\ne(S)$ from \eqref{eqn:J_ne} 
and their $\nwarrow$, $\swarrow$, $\searrow$ variants.
We define the \textit{$\Theta$-future} of a point $p\in N$ as the subset
\begin{subequations}
\begin{flalign}
J_\Theta^\uparrow(p) \,:=\,
\big\{q \in N \, \big\vert \,
\exists s \geq 0 \text{ s.t. } q = \Theta(s,p)
\big\}\,\subseteq\,N
\end{flalign}
and the \textit{$\Theta$-past} of $p\in N$ as the subset
\begin{flalign}
J_\Theta^\downarrow(p) \,:=\,
\big\{q \in N \, \big\vert \,
\exists s \leq 0 \text{ s.t. } q = \Theta(s,p)
\big\} \,\subseteq\,N\quad.
\end{flalign}
\end{subequations}
Recalling Remark \ref{rem:J_from_flow},
we observe that $J^\uparrow_{\Theta_+}(p) = J^\ne(p)$ and $J^\downarrow_{\Theta_+}(p) = J^\sw(p)$ 
for $\Theta_+$ a $+$-chiral flow on $N = \partial M$, and 
$J^\uparrow_{\Theta_-}(p) = J^\nw(p)$ and $J^\downarrow_{\Theta_-}(p) = J^\se(p)$ for $\Theta_-$ a $-$-chiral flow.
We also define for any subset 
$S \subseteq N$ the $\Theta$-future/past of $S$ to be
\begin{flalign} \label{eqn:Theta-past-future-sets}
J_\Theta^{\uparrow/ \downarrow}(S) \,:=\,
\bigcup_{p \in S} J_\Theta^{\uparrow /\downarrow}(p) \,\subseteq\, N
\quad.
\end{flalign}
With this we can generalize Definition \ref{def:Greenshomotopies_on_partial_M} of Green's homotopies 
for $\mathfrak{L}^\pm$ to the $k$-shifted de Rham complex $\Omega^\bullet(N)[k]$ on $N$.
\begin{defi} \label{def:Greenshomotopies}
Let $N$ be a smooth manifold (with or without boundary) that is 
endowed with an $\bbR$-action $\Theta : \bbR \times N \to N$.
A \textit{forward/backward Green's homotopy} 
for the $k$-shifted de Rham complex $\Omega^\bullet(N)[k]$
is a $(-1)$-cochain $G^{\uparrow/\downarrow}\in \big[ \Omega^\bullet_\cc(N)[k], \Omega^\bullet(N)[k] \big]^{-1}$
in the internal hom complex that satisfies the following properties:
\begin{enumerate}[(i)]
\item \label{def:Greenshomotopies:eom}
$\partial G^{\uparrow /\downarrow} = j$, where $j: \Omega^\bullet_\cc(N)[k] \to \Omega^\bullet(N)[k]$
is the canonical inclusion cochain map from compactly supported forms to all forms;

\item \label{def:Greenshomotopies:support}
$\supp\big(G^{\uparrow/ \downarrow} \alpha \big)\subseteq
J_\Theta^{\uparrow/ \downarrow} \big(\supp \,\alpha\big)$, for all
$\alpha \in \Omega^\bullet_\cc(N)[k]$.
\end{enumerate}
\end{defi}

\begin{rem} \label{rem:Greenshomotopies:index_shifting}
Note that a forward/backward Green's homotopy 
$G^{\uparrow/\downarrow} \in \big[ \Omega^\bullet_\cc(N), \Omega^\bullet(N) \big]^{-1}$
for the \textit{unshifted} de Rham complex induces corresponding forward/backward Green's homotopies 
$G^{\uparrow/\downarrow}_{[k]} := (-1)^k\,G^{\uparrow /\downarrow}\in 
\big[ \Omega^\bullet_\cc(N)[k], \Omega^\bullet(N)[k] \big]^{-1}$ 
for each $k$-shifted de Rham complex. The sign factor
ensures that property \ref{def:Greenshomotopies:eom}, i.e.\ $\partial G^{\uparrow/\downarrow}_{[k]} =j$, holds 
true for the internal hom differential $\partial$ associated with the $k$-shifted
de Rham differential $\dd_{[k]} = (-1)^k\, \dd$.
For this reason, we construct our candidate Green's homotopies below only for the unshifted de Rham complex.
\end{rem}

Let us consider the following span of smooth maps
\begin{flalign} \label{eqn:fibre-integration-span}
\begin{gathered}
\xymatrix{
&\ar[dl]_-{\pr}  \bbR\times N\ar[dr]^-{\Theta}&\\
N&&N
}
\end{gathered}
\end{flalign}
where $\pr$ denotes the projection onto the second factor.
We will construct Green's homotopies in the sense of 
Definition \ref{def:Greenshomotopies} by first pulling forms back along $\Theta$
and then integrating over (suitable subsets of) 
the fibers of $\mathrm{pr}$.
The latter is a trivial bundle with oriented fiber, and so is oriented as a bundle.
To ensure that the composite operation
is well-defined on compactly supported forms, we have to impose
a niceness condition on the $\bbR$-action $\Theta$, namely that it is proper.
Let us start by recalling that, for a smooth fiber bundle $\pi : E \to B$ with fibers of dimension $r$,
fiber integration  $\pi_\ast : \Omega^k_\mathrm{vc}(E) \to \Omega^{k-r}(B)$ is generally only well-defined on
forms $\alpha \in \Omega^k_{\mathrm{vc}}(E)\subseteq \Omega^k(E)$ 
with vertically compact support. This support condition
means that, given any compact subset $K\subseteq B$ in the base,
the subset $\pi^{-1}(K) \cap\, \supp \,\alpha \subseteq E$ is compact in the total space.
For standard references on fiber integration, see e.g.\ \cite[Chapter VII]{GHV}
as well as \cite[Section 6]{BottTu} and \cite[Section 3.4.5]{Nicolaescu}.\footnote{We 
use the fiber-first orientation convention for fiber integration, as in \cite{Nicolaescu} but in contrast to \cite{GHV,BottTu}.
In particular, fiber integration and the de Rham differential commute up to a sign
$\pi_\ast \circ \dd = (-1)^r \,\dd \circ \pi_\ast$
when the fibers have empty boundary.}
\begin{lem} \label{lem:proper-implies-vertically-compact}
Suppose that the $\bbR$-action $\Theta:\bbR\times N\to N$ is proper.
Then, for each compactly supported form  $\alpha \in \Omega^k_\cc(N)$,
the pulled back form $\Theta^\ast \alpha\in \Omega^k(\bbR\times N)$ 
has vertically compact support with respect to the trivial
$\bbR$-bundle $\pr : \bbR \times N \to N$.
\end{lem}
\begin{proof}
Recall the definition of the shear map
\begin{flalign} \label{eqn:shear_map}
\theta \,:\, \bbR \times N \,\longrightarrow\, N \times N~,~~ (s, p) \, \longmapsto\, \big(p, \Theta(s,p)\big)
\end{flalign}
associated with the $\bbR$-action $\Theta$
and observe that
$\theta^{-1}(A \times B)=
\mathrm{pr}^{-1}(A) \cap \Theta^{-1}(B)$
for any two subsets $A,B \subseteq N$.
Consider any compact subset $K \subseteq N$.
Since $\supp (\Theta^{\ast} \alpha) \subseteq \Theta^{-1} \big(\supp\,\alpha\big)$
and supports are closed sets, it follows that
\begin{flalign}
\mathrm{pr}^{-1}(K) \cap \supp(\Theta^{\ast}\alpha)\,\subseteq\,
\mathrm{pr}^{-1}(K) \cap \Theta^{-1}\big(\supp\,\alpha\big)
\,=\, \theta^{-1}\big(K \times \supp\,\alpha\big)
\end{flalign}
is a closed subset. Using now that $\Theta$ is a proper $\bbR$-action,
i.e.\ its shear map $\theta$ is a proper map, it follows
that $\mathrm{pr}^{-1}(K) \cap \supp( \Theta^{\ast}\alpha)\subseteq \theta^{-1}\big(K \times \supp\,\alpha\big)$
is a closed subset of a compact set, and hence is compact.
\end{proof}

To propose a candidate for a forward Green's homotopy, 
we restrict the $\bbR$-action to non-positive parameters
$\bbR^{\leq 0}\subseteq \bbR$
and consider the corresponding restriction
$N \stackrel{\pr}{\longleftarrow} \bbR^{\leq 0}\times N \stackrel{\Theta}{\longrightarrow} N$
of the span \eqref{eqn:fibre-integration-span}.
By an obvious corollary of Lemma \ref{lem:proper-implies-vertically-compact}, 
if $\Theta$ is a proper $\bbR$-action, we obtain a family of linear maps
\begin{flalign} \label{eqn:fibre-integration-dR}
\xymatrix{
\Omega_\cc^k(N)
\ar[r]^-{\Theta^\ast}
~&~
\Omega_{\vc}^k(\bbR^{\leq 0}\times N)
\ar[r]^-{\int_{\bbR^{\leq 0}}}
~&~
\Omega^{k-1}(N)
}
\end{flalign}
by composing the pullback of forms along $\Theta$ with the
fiber integration $\int_{\bbR^{\leq 0}} := \pr_\ast$ along $\pr$.
These components define the $(-1)$-cochain
\begin{flalign}\label{eqn:Gforward}
G^\uparrow_\Theta \,:=\, \int_{\bbR^{\leq 0}}\Theta^\ast \,\in\, \big[ \Omega^\bullet_\cc(N), \Omega^\bullet(N)\big]^{-1}\quad.
\end{flalign}
Restricting to non-negative parameters, we obtain again a restricted 
span $N \stackrel{\pr}{\longleftarrow} \bbR^{\geq 0}\times N \stackrel{\Theta}{\longrightarrow} N$ 
and a family of linear maps
\begin{flalign}\label{eqn:fibre-integration-dR2}
\xymatrix{
\Omega_\cc^k(N)
\ar[r]^-{\Theta^\ast}
~&~
\Omega_{\vc}^k(\bbR^{\geq 0}\times N)
\ar[r]^-{\int_{\bbR^{\geq 0}}}
~&~
\Omega^{k-1}(N)\quad.
}
\end{flalign}
Taking the additive inverses of these components defines the $(-1)$-cochain
\begin{flalign}\label{eqn:Gbackward}
G^\downarrow_\Theta \,:=\, - \int_{\bbR^{\geq 0}}\Theta^\ast \,\in\, \big[ \Omega^\bullet_\cc(N), \Omega^\bullet(N)\big]^{-1}\quad.
\end{flalign}
Beyond well-definedness of \eqref{eqn:Gforward} and \eqref{eqn:Gbackward},
the assumption that $\Theta$ is proper gives
some further technical consequences described in Appendix \ref{app:appendix}.
These enter the proof of the following
\begin{propo}\label{propo:Gexistence}
Suppose that the $\bbR$-action $\Theta:\bbR\times N\to N$ is proper.
Then $G^{\uparrow}_\Theta$ in \eqref{eqn:Gforward} 
defines a forward Green's homotopy and $G^{\downarrow}_\Theta$ in \eqref{eqn:Gbackward} 
defines a backward Green's homotopy.
\end{propo}
\begin{proof}
We have to verify the two properties \ref{def:Greenshomotopies:eom}  and \ref{def:Greenshomotopies:support} 
from Definition \ref{def:Greenshomotopies}.
\sk

Item \ref{def:Greenshomotopies:eom} is a simple consequence 
of the fiber-wise Stokes theorem, see e.g.\ \cite[Theorem 3.4.54]{Nicolaescu}.
For the $(-1)$-cochain $G^{\uparrow}_\Theta\in \big[ \Omega^\bullet_\cc(N), \Omega^\bullet(N)\big]^{-1}$, 
we compute
\begin{flalign}\label{eqn:partialG}
\partial G^\uparrow_\Theta = 
\dd\, \int_{\mathbb{R}^{\leq 0}} \Theta^\ast + \int_{\mathbb{R}^{\leq 0}} \Theta^\ast\, \dd =
\left( \dd \int_{\mathbb{R}^{\leq 0}} + \int_{\mathbb{R}^{\leq 0}} \dd \right) \,\Theta^\ast =
\int_{\partial \mathbb{R}^{\leq 0}} \Theta^\ast\quad,
\end{flalign}
where in the second step we used that the de Rham differential 
commutes with pullbacks of differential forms and in the third step we used the
fiber-wise Stokes theorem. Since the boundary $\partial \bbR^{\leq 0}=\{0\}$ is a singleton, the
fiber integration $\int_{\partial \mathbb{R}^{\leq 0}}$ is given by 
the pullback of forms along $N\to \bbR^{\leq 0}\times N\, ,~p\mapsto (0,p)$.
Since $\Theta(0,\,\cdot\,) = \id_N$ yields the identity, this
shows that $\partial G^\uparrow_\Theta = j : \Omega^\bullet_\cc(N)\to \Omega^\bullet(N)$
is the canonical inclusion. The proof for $\partial G^\downarrow_\Theta = j$ is similar 
and uses that the boundary $\partial \bbR^{\geq 0}=\{0\}$ is negatively oriented, which accounts for
the minus sign in \eqref{eqn:Gbackward}.
\sk

To prove item \ref{def:Greenshomotopies:support}, 
let us first note that the subset $J^{\uparrow/\downarrow}_\Theta(K)\subseteq N$ 
is closed for each compact subset $K \subseteq N$
because the $\mathbb{R}$-action $\Theta$ is proper, see Corollary \ref{cor:Theta_future_past_of_compact_is_closed}.
Let us take any compactly supported
form $\alpha\in \Omega^k_\cc(N)$ and consider $G^{\uparrow}_\Theta \alpha\in\Omega^{k-1}(N)$.
Recall that $\supp\big(G^{\uparrow}_\Theta \alpha\big)\subseteq N$ is defined
as the closure of the set of points $p\in N$ for which the evaluation
$\big(G^{\uparrow}_\Theta \alpha\big)(p)\neq 0$ is non-vanishing.
It suffices to show that each such point satisfies 
$p\in J^\uparrow_\Theta\big(\supp\,\alpha\big)$ as the latter set is closed.
From $\big(G^{\uparrow}_\Theta \alpha\big)(p) = \big(\int_{\bbR^{\leq 0}}\Theta^\ast \alpha\big)(p)\neq 0$, 
it follows that there exists a parameter $s\leq 0$ such that $\big(\Theta^\ast \alpha\big)(s,p)\neq 0$,
hence one finds a point $(s,p)\in \supp(\Theta^\ast\alpha)\subseteq \Theta^{-1}(\supp\,\alpha)$.
Since $\Theta(s,p) \in \supp\, \alpha$ for non-positive $s \leq 0$, 
we conclude that $p\in J^\uparrow_\Theta\big(\supp\,\alpha\big)$.
By a similar argument one shows that $\supp \big(G^{\downarrow}_\Theta\alpha\big) 
\subseteq J^{\downarrow}_\Theta\big(\supp\, \alpha\big)$.
\end{proof}

\paragraph{$G^{\uparrow / \downarrow}_\Theta$ as restricted fiber integrations:}
Our Green's homotopies $G^{\uparrow / \downarrow}_\Theta \in 
\left[ \Omega^\bullet_\cc(N), \Omega^\bullet(N)\right]^{-1}$
may be understood as integration over ``backward''/``forward'' halves
of the fibers of the quotient $\pi : N \to N /\mathbb{R}$
of $N$ by the $\mathbb{R}$-action $\Theta$.
The latter is a principal $\mathbb{R}$-bundle by Corollary \ref{cor:proper_flow_gives_principal_bundle}
because the $\mathbb{R}$-action $\Theta$ is proper.
We now demonstrate this interpretation using appropriate coordinate expressions.
\sk

Denote $B := N /\mathbb{R}$.
We work in a trivialization $N \cong \mathbb{R} \times B$,
which exists since principal $\mathbb{R}$-bundles are trivializable owing to the contractibility of $\mathbb{R}$.
Forms on the product space decompose
$\Omega^k_\mathrm{c}(\mathbb{R} \times B) =
\Omega^{1,k-1}_\mathrm{c}(\mathbb{R} \times B) \oplus \Omega^{0,k}_\mathrm{c}(\mathbb{R} \times B)$
into forms with one leg along the $\mathbb{R}$-factor and forms with no legs along $\mathbb{R}$.
Pick a global coordinate $\tau$ on the $\mathbb{R}$-factor.
Since trivializations of principal bundles are equivariant,
the $\mathbb{R}$-action $\Theta$ becomes translation along the $\mathbb{R}$-factor
$\Theta(s, (\tau,p)) = (\tau +s, p)$, for $p \in B$ and group element $s \in \mathbb{R}$.
Forms $\phi \in \Omega^{1,k-1}_\mathrm{c}(\mathbb{R} \times B)$ and
$\psi \in \Omega^{0,k}_\mathrm{c}(\mathbb{R} \times B)$ may be written as
\begin{flalign} \label{eqn:forms_split_coordinates}
\phi{(\tau,p)} = f(\tau,p) \, \dd \tau \wedge \widetilde{\phi}(p)
\qquad \text{and} \qquad
\psi{(\tau,p)} = g(\tau,p) \, \widetilde{\psi}(p) \quad,
\end{flalign}
for $\widetilde{\phi} \in \Omega^{k-1}(B)$ and $\widetilde{\psi} \in \Omega^k(B)$ respectively,
with $f,g \in C^\infty_\mathrm{c}(\mathbb{R} \times B)$.
We leave implicit the pullbacks along the projections out of the product $\mathbb{R} \times B$.
It follows that
\begin{flalign}
\nn (\Theta^\ast \phi){(s,(\tau,p))} &= f(\tau + s,p) \, (\dd \tau + \dd s)\wedge \widetilde{\phi}(p)\\
(\Theta^\ast \psi){(s,(\tau,p))} &= g(\tau + s,p) \, \widetilde{\psi}(p) \quad.
\end{flalign}
Since the fiber integrations in $G^{\uparrow}_\Theta = \int_{\mathbb{R}^{\leq 0}} \Theta^*$
and $G^{\downarrow}_\Theta = -\int_{\mathbb{R}^{\geq 0}} \Theta^*$
are along the fiber direction $\dd s$ only,
it is immediate that $G^{\uparrow / \downarrow}_\Theta \psi = 0$
for $\psi \in \Omega^{0,k}_\mathrm{c}(\mathbb{R} \times B)$ of the second type.
For forms $\phi \in \Omega^{1,k-1}_\mathrm{c}(\mathbb{R} \times B)$ of the first type, we find
\begin{subequations}
\begin{flalign}
(G^\uparrow_\Theta \phi){(\tau, p)}
 \,=\, \left(\int_{-\infty}^0
f(\tau + s,p) \,\dd s\right) \, \widetilde{\phi}(p)
\,=\, \left(\int_{-\infty}^\tau f(\tau^\prime,p) \,\dd \tau^\prime\right)\, \widetilde{\phi}(p) \quad,
\label{eqn:G^up_Theta_coordinates}
\intertext{and similarly}
(G^\downarrow_\Theta \phi){(\tau, p)}
 \,=\, -\left(\int_{0}^\infty
f(\tau + s,p)\, \dd s\right)\, \widetilde{\phi}(p)
\,=\, -\left(\int_\tau^\infty f(\tau^\prime,p)\, \dd \tau^\prime\right) \, \widetilde{\phi}(p) \quad.
\label{eqn:G^down_Theta_coordinates}
\end{flalign}
\end{subequations}
Recall that the fiber integration $\pi_\ast : \Omega^\bullet_\mathrm{vc}(N) \to \Omega^{\bullet-1}(B)$
over the bundle $\pi : N \to B$
(canonically oriented as per Lemma \ref{lem:principal_R-bundles_are_oriented})
has a similar expression
\begin{flalign}
(\pi_\ast \phi)(p) \, =\, \left(\int_{-\infty}^\infty
f(\tau^\prime,p) \, \dd \tau^\prime\right)\,  \widetilde{\phi}(p)\quad,
\end{flalign}
for $\phi \in \Omega^{1,k-1}_\mathrm{c}(\mathbb{R} \times B)$ of the first type
and $\pi_\ast \psi = 0$ for $\psi \in \Omega^{0,k}_\mathrm{c}(\mathbb{R} \times B)$ of the second type.
The Green's homotopies $G^{\uparrow / \downarrow}_\Theta$ are thus
appropriately signed \emph{restricted fiber integrations} along the orbits of $\Theta$.
The following is an immediate consequence.
\begin{propo} \label{prop:G_Theta:difference}
The Green's homotopies $G^{\uparrow / \downarrow}_\Theta \in [\Omega^\bullet_\mathrm{c}(N), \Omega^\bullet(N)]^{-1}$
of \eqref{eqn:Gforward} and \eqref{eqn:Gbackward} satisfy
\begin{flalign}
G^\uparrow_\Theta - G^\downarrow_\Theta
= \pi^\ast\, \pi_\ast \,:\, \Omega^\bullet_\mathrm{c} (N) \,\longrightarrow\, \Omega^{\bullet-1}(N)\quad,
\end{flalign}
where $\pi : N \to N/\mathbb{R}$ is the quotient by the proper $\mathbb{R}$-action $\Theta$.
\end{propo}

\subsection{\label{subsec:GH_on_boundary_and_bulk}Boundary and bulk Green's homotopies}
We now specialize the Green's homotopies from Subsection
\ref{subsec:Greens_homotopies_from_flow}, first to the self-dual and anti-self-dual boundary condition complexes
$\mathfrak{L}^+$ and $\mathfrak{L}^-$ and finally to the complex $\mathfrak{F}_{\mathfrak{L}^\pm}(M)$ 
of boundary conditioned fields of linear Chern-Simons theory on $M$.

\paragraph{Boundary Green's homotopies:} For our construction of the Green's homotopies 
\eqref{eqn:Gforward} and \eqref{eqn:Gbackward} for the de Rham complex,
it was crucial to assume that the $\bbR$-action is proper. 
Since the chiral flows $\Theta_+, \Theta_- : \bbR\times \partial M\to \partial M$ 
on a $2$-dimensional oriented and time-oriented Lorentzian manifold $\partial M$
are in general not proper, we have to assume the following strong chiral niceness condition.
\begin{assu}\label{assu:properchiralflows}
The boundary Lorentzian manifold $\partial M$ admits a proper chiral flow $\Theta_\pm$.
\end{assu}

\begin{rem} \label{rem:proper-flow-Cauchy-slice}
If $\Theta_\pm$ is proper, then the associated quotient map 
$\pi_\pm : \partial M \to \partial M /\!\sim_\pm $ has the structure of a
smooth principal $\bbR$-bundle, see Corollary \ref{cor:proper_flow_gives_principal_bundle}.
Because $\bbR$ is contractible, the quotient map $\pi_\pm$ admits smooth global sections,
which can be used to define a chiral analogue of Cauchy surfaces in $\partial M$ \cite[Section 4.2]{Grant-Stuart}.
Properness of the chosen chiral flow $\Theta_\pm$ may therefore be understood as a chiral analogue of global hyperbolicity.
\end{rem}

\begin{ex}
If $\partial M$ is globally hyperbolic (in the usual sense),
then any $+$- or $-$-chiral flow on it is proper.
Hence, Assumption \ref{assu:properchiralflows} holds true
whenever $\partial M$ is globally hyperbolic.
This gives a large class of examples satisfying our hypotheses.
\end{ex}

In order to prove that the Green's homotopies from Proposition 
\ref{propo:Gexistence} restrict to the boundary condition $\LLL^\pm\subseteq \Omega^\bullet(\partial M)[1]$
we require the following technical lemma.
\begin{lem} \label{lem:invariant_self-dual_form}
Suppose that $\Theta_+ : \bbR \times \partial M \to \partial M$ is a proper $+$-chiral flow.
Then there exists a nowhere-vanishing self-dual form
$\beta^+ \in \Omega^+(\partial M)$ that is invariant under $\Theta_+$
in the sense that
\begin{flalign}
\Theta_+^\ast\beta^+ \,= \,\pr^\ast \beta^+
\,\in\, \Omega^1(\bbR\times \partial M)\quad,
\end{flalign}
where $\pr : \bbR \times \partial M\to \partial M$ denotes the projection map.
\sk

Similarly, if $\Theta_- : \bbR \times \partial M \to \partial M$ is a 
proper $-$-chiral flow then there exists a nowhere-vanishing anti-self-dual form
$\beta^-\in\Omega^-(\partial M)$ satisfying $\Theta_-^\ast\beta^- =\pr^\ast\beta^-$.
\end{lem}
\begin{proof}
It suffices to prove the statement for $\Theta_+$ and self-dual forms
because the other case is similar.
Using Corollary \ref{cor:proper_flow_gives_principal_bundle},
one has that the quotient $\pi_+ :\partial M \to \partial M/\!\sim_+$
associated with the $\bbR$-action $\Theta_+$ carries the structure of a smooth principal $\bbR$-bundle.
Choosing any nowhere-vanishing $1$-form $\lambda\in \Omega^1(\partial M/\!\sim_+)$ on the base,
which exists because $\partial M/\!\sim_+$ is a $1$-manifold and hence orientable, we define
the $1$-form $\beta^+  := \pi^\ast_+\lambda \in \Omega^1(\partial M)$.
This form satisfies $\Theta_+^\ast\beta^+ \,= \,\pr^\ast\beta^+$ because
the $\bbR$-action $\Theta_+$ preserves the fibers of $\pi_+:\partial M\to \partial M/\!\sim_+$, 
i.e.\ $\pi_+\circ\Theta_+ = \pi_+ \circ \pr$.
By construction, the $1$-form $\beta^+  = \pi^\ast_+\lambda$ annihilates all 
vertical tangent vectors of the bundle $\pi_+:\partial M\to \partial M/\!\sim_+$,
which include all $\nearrow$- and $\swarrow$-pointing tangent vectors on $\partial M$.
Hence, $\beta^+$ is self-dual by Lemma \ref{lem:sd/asdannihilation}.
\end{proof}

\begin{ex}
Consider the $2$-dimensional Minkowski spacetime
with $+$- and $-$-chiral flows given by
$\Theta_+ \left(s, (x^-, x^+) \right) = (x^-, x^+ + s)$ and 
$\Theta_- \left(s, (x^-, x^+) \right) = (x^- + s, x^+)$ in light cone coordinates,
as per Examples \ref{ex:Minkowski1}, \ref{ex:Minkowski2} and \ref{ex:Minkowski:flows}.
Then $\beta^+ = \dd x^-$ is a $\Theta_+$-invariant self-dual form
and $\beta^- = \dd x^+$ is a $\Theta_-$-invariant anti-self-dual form.
As remarked already in Example \ref{ex:Minkowski1}, note the sign flip.
\end{ex}

\begin{propo}\label{prop:G_Theta:restricts_to_L}
Suppose that $\partial M$ has a proper $+$-chiral flow $\Theta_+$.
Then the self-dual boundary condition $\LLL^+$
admits a $\nearrow$-pointed Green's homotopy $G^\ne$ and a $\swarrow$-pointed Green's homotopy
$G^\sw$ in the sense of Definition \ref{def:Greenshomotopies_on_partial_M}.
An explicit choice is given by restricting the Green's homotopies
\begin{flalign}
G^{\uparrow/\downarrow}_{\Theta_+ [1]} \,=\,-G^{\uparrow/\downarrow}_{\Theta_+} \,\in\,\big[\Omega^\bullet_\cc(\partial M)[1], \Omega^\bullet(\partial M)[1]\big]^{-1}
\end{flalign}
from Proposition \ref{propo:Gexistence} and Remark \ref{rem:Greenshomotopies:index_shifting}
to the boundary condition $\LLL^+\subseteq \Omega^\bullet(\partial M)[1]$.
\sk

Similarly, if $\partial M$ has a proper $-$-chiral flow $\Theta_-$
then the anti-self-dual boundary condition $\LLL^-$
admits a $\nwarrow$-pointed Green's homotopy $G^\nw$ and a $\searrow$-pointed Green's homotopy $G^\se$,
with an explicit choice given by restricting
\begin{flalign}
G^{\uparrow/\downarrow}_{\Theta_- [1]} \,=\,-G^{\uparrow/\downarrow}_{\Theta_-} \,\in\,\big[\Omega^\bullet_\cc(\partial M)[1], \Omega^\bullet(\partial M)[1]\big]^{-1}
\end{flalign}
to the boundary condition $\LLL^-\subseteq \Omega^\bullet(\partial M)[1]$.
\end{propo}
\begin{proof}
The only nontrivial step is to show that the restrictions exist.
We only spell out the proof for the case of $\LLL^+$ and $G^{\uparrow}_{\Theta_+[1]}$ 
since all other cases are similar. Recalling the explicit form of the boundary condition complex 
\eqref{eqn:boundarycondition}, we have to show that,
for every compactly supported self-dual $1$-form $\alpha\in \Omega_\cc^+(\partial M)$,
the $0$-form $-G^{\uparrow}_{\Theta_+}(\alpha) \in\Omega^0(\partial M)$ is zero
and that, for every compactly supported $2$-form $\alpha\in \Omega_\cc^2(\partial M)$,
the $1$-form $-G^{\uparrow}_{\Theta_+}(\alpha) \in\Omega^1(\partial M)$ is self-dual.
Using the $\Theta_+$-invariant module basis from Lemma \ref{lem:invariant_self-dual_form},
we can write $\alpha = \phi \wedge \beta^+$ 
where $\phi \in C^\infty_{\cc} (\partial M)$ is a compactly supported function
in the first case and $\phi \in \Omega^-_{\cc}(\partial M)$ is a compactly supported
anti-self-dual form in the second case. Using \eqref{eqn:Gforward}, we compute
\begin{flalign}
 -G^{\uparrow}_{\Theta_+}(\alpha) &= -\int_{\bbR^{\leq 0}} \Theta_+^\ast(\phi \wedge \beta^+)
=-\int_{\bbR^{\leq 0}} \big(\Theta_+^\ast(\phi) \wedge \pr^\ast(\beta^+)\big)
=-\left(\int_{\bbR^{\leq 0}} \Theta_+^\ast(\phi)\right) \wedge \beta^+ \; \, ,
\end{flalign}
where in the second step we used that $\beta^+$ is $\Theta_+$-invariant (see Lemma \ref{lem:invariant_self-dual_form})
and the last step is an application of the projection formula for fiber integration 
\cite[Equation (3.4.11)]{Nicolaescu}. In the first case $\phi$ is a $0$-form, so its fiber 
integral over $1$-dimensional fibers vanishes. It follows that $-G^\uparrow_{\Theta_+}(\alpha) = 0$.
In the second case $\phi$ is a $1$-form, so its fiber 
integral is a function. Since $\beta^+$ is self-dual, 
it follows that $-G^{\uparrow}_{\Theta_+}(\alpha) \in\Omega^+(\partial M)$.
This proves that $-G^{\uparrow}_{\Theta_+}$ restricts to 
$G^\ne:=-G^{\uparrow}_{\Theta_+}\in [\LLL_\cc^+,\LLL^+]^{-1} $.
\end{proof}

\paragraph{Bulk Green's homotopies:}
We would like to extend the boundary Green's homotopies from Proposition 
\ref{prop:G_Theta:restricts_to_L} to the complex $\FFF_{\LLL^\pm}(M)\subseteq \Omega^\bullet(M)[1]$
of boundary conditioned Chern-Simons fields on $M$. Such extensions
can be obtained from the techniques developed in Subsection \ref{subsec:Greens_homotopies_from_flow},
provided that the following assumption holds true.
\begin{assu}\label{assu:bulkflows}
The bulk manifold $M$ admits a proper $\bbR$-action $\widehat{\Theta}_{\pm} : \bbR\times M\to M$
that extends the chosen chiral flow $\Theta_\pm$ on the boundary $\partial M$,
i.e.\ the diagram
\begin{flalign}\label{eqn:bulkflowcompatibility}
\begin{gathered}
\xymatrix{
\bbR\times M \ar[r]^-{\widehat{\Theta}_{\pm}} & M\\
\ar[u]^-{\id\times\iota}\bbR\times \partial M \ar[r]_-{\Theta_\pm} & \partial M\ar[u]_-{\iota}
}
\end{gathered}
\end{flalign}
commutes, where $\iota:\partial M\to M$ denotes the boundary inclusion.
\end{assu}

\begin{rem}
The extension $\widehat{\Theta}_\pm$ of the chosen chiral flow $\Theta_\pm$
required by Assumption \ref{assu:bulkflows} is clearly \textit{not} unique.
In Section \ref{sec:AQFTonM}, we will construct an AQFT for the Chern-Simons/Wess-Zumino-Witten bulk/boundary system
out of $\widehat{\Theta}_\pm$ and the associated Green's homotopies
from Proposition \ref{prop:G_Theta:restricts_to_FL} below.
This prompts the question: How canonical is our AQFT construction?
\sk

As we will show in Section \ref{sec:dimred},
the AQFT for this bulk/boundary system
is (up to weak equivalence) only sensitive to the base manifold
of the quotient map $\widehat{\pi}_\pm : M\to M/\!\sim_\pm $ associated
with the $\bbR$-action $\widehat{\Theta}_\pm$.
Under certain additional topological assumptions on $M$,
we show in Proposition \ref{prop:diffeomorphic-bases} that any
two such base manifolds arising from different choices of
$\widehat{\Theta}_\pm$ must be diffeomorphic.
From this, it follows that any two proper extensions of
the chiral flow $\Theta_\pm$ to $M$ yield weakly equivalent AQFTs.
\end{rem}

\begin{rem}
Commutativity of \eqref{eqn:bulkflowcompatibility} implies 
in particular that $\partial M$ is a $\widehat{\Theta}_\pm$-invariant subspace of $M$.
It is a general fact of continuous group actions that if $\widehat{\Theta}_\pm$ 
is proper, then its restriction to any invariant subspace is also proper.
Consequently, Assumption \ref{assu:bulkflows} implies Assumption 
\ref{assu:properchiralflows} that the chosen chiral flow $\Theta_\pm$ on the boundary is proper.
\end{rem}

\begin{ex}\label{ex:Minkowski3}
Continuing Examples \ref{ex:Minkowski1}, \ref{ex:Minkowski2} and \ref{ex:Minkowski:flows},
consider $M = \bbR^2\times \bbR^{\geq 0}$ with boundary given by the standard Minkowski spacetime.
Then the proper $\bbR$-action
\begin{flalign}
\bbR\times M\, \longrightarrow\,M~,~~
\big(s,(t,x,r)\big)\,\longmapsto\,
\big(t+\tfrac{s}{2}, x + \epsilon \,\tfrac{s}{2},r\big)
\end{flalign}
provides an extension of the chiral flows $\Theta_+$ and $\Theta_-$ 
from Example \ref{ex:Minkowski:flows} when $\epsilon = +1$ and $\epsilon = -1$ respectively.
\end{ex}

\begin{ex}\label{ex:cylinder}
As a second more interesting example, consider the filled cylinder $M = \bbR\times \mathbb{D}^2$
with $\mathbb{D}^2 := \{(x,y)\in \bbR^2 \,\vert\, x^2 +y^2 \leq 1\}$ the closed unit disk. 
Working in polar coordinates $(r,\phi)\in (0,1]\times \bbR/\bbZ$ on the disk, we endow $M$ with the orientation
$\mathfrak{o} = [\dd t\wedge\dd \phi \wedge \dd r]$
and its boundary cylinder $\partial M $ (defined by $r=1$)
with the flat metric $g = -\dd t^2 + \dd \phi^2$ and time-orientation $\mathfrak{t}=[\partial_t]$.
The proper $\bbR$-action
\begin{flalign}
\bbR\times M\, \longrightarrow\,M~,~~ \big(s,(t,\phi,r)\big)\,\longmapsto\,
\big(t+\tfrac{s}{2}, \phi + \epsilon\, r\,\tfrac{s}{2},r\big)
\end{flalign}
provides an extension of the flow 
$(s,(t,\phi)) \mapsto \big(t+\tfrac{s}{2}, \phi + \epsilon \,\tfrac{s}{2}\big)$
on the flat Lorentzian cylinder. The latter is $+$-chiral when $\epsilon = +1$ 
and $-$-chiral when $\epsilon = -1$.
\end{ex}

\begin{propo}\label{prop:G_Theta:restricts_to_FL}
Suppose that Assumption \ref{assu:bulkflows} holds.
Then the forward/backward Green's homotopies
\begin{flalign}
G^{\uparrow/\downarrow}_{\widehat{\Theta}_\pm[1]}\,=\,-G^{\uparrow/\downarrow}_{\widehat{\Theta}_\pm} \,\in\,\big[\Omega^\bullet_\cc(M)[1], \Omega^\bullet(M)[1]\big]^{-1}
\end{flalign}
from Proposition \ref{propo:Gexistence} and Remark \ref{rem:Greenshomotopies:index_shifting}
restrict to the complex $\FFF_{\LLL^\pm}(M)\subseteq \Omega^\bullet(M)[1]$ of boundary conditioned fields.
These bulk Green's homotopies are compatible
with the boundary ones from Proposition \ref{prop:G_Theta:restricts_to_L} in the sense that
\begin{flalign}\label{eqn:G-bulk-boundary-compatibility}
\iota^\ast\, G^{\uparrow/\downarrow}_{\widehat{\Theta}_\pm[1]} = G^{\uparrow/\downarrow}_{\Theta_\pm[1]}\,\iota^\ast
\quad,
\end{flalign}
where $\iota : \partial M\to M$ denotes the boundary inclusion.
\end{propo}
\begin{proof}
Recalling the definition of $\FFF_{\LLL^\pm}(M)\subseteq \Omega^\bullet(M)[1]$
from \eqref{eqn:conditionedfieldspullback}, we observe that 
the compatibility conditions \eqref{eqn:G-bulk-boundary-compatibility}
together with Proposition \ref{prop:G_Theta:restricts_to_L} would imply that the desired restrictions exist. 
Indeed, we would have to show that
$\iota^\ast G^{\uparrow/\downarrow}_{\widehat{\Theta}_\pm[1]}(\alpha)\in \LLL^\pm$,
for all $\alpha\in \FFF_{\LLL^\pm,\cc}(M)\subseteq \Omega_\cc^\bullet(M)[1]$,
i.e.\ $\iota^\ast \alpha \in \LLL^\pm_{\cc}$.
Compatibility gives that
$\iota^\ast G^{\uparrow/\downarrow}_{\widehat{\Theta}_\pm[1]}(\alpha) = 
G^{\uparrow/\downarrow}_{\Theta_\pm[1]}\big(\iota^\ast \alpha\big)$,
which then lies in $\mathfrak{L}^\pm$ by Proposition \ref{prop:G_Theta:restricts_to_L}.
\sk

It thus remains to prove \eqref{eqn:G-bulk-boundary-compatibility}, for which it suffices to consider
the case of $G^{\uparrow}_{\widehat{\Theta}_\pm}$ since the other case of $G^{\downarrow}_{\widehat{\Theta}_\pm}$ is similar.
The boundary inclusion $\iota : \partial M\to M$ extends to a bundle morphism
\begin{flalign}
\begin{gathered}
\xymatrix{
\ar[d]_-{\pr}\bbR^{\leq 0}\times \partial M \ar[r]^-{\id\times\iota} & \bbR^{\leq 0}\times M\ar[d]^-{\pr}\\
\partial M\ar[r]_-{\iota}& M
}
\end{gathered}
\end{flalign}
which is fiber-wise a diffeomorphism. Recalling \eqref{eqn:Gforward}, we compute
\begin{flalign}
\iota^\ast\, G^{\uparrow}_{\widehat{\Theta}_\pm[1]}
= -\iota^\ast\, \int_{\bbR^{\leq 0}}\widehat{\Theta}^\ast_\pm 
= -\int_{\bbR^{\leq 0}}(\id\times\iota)^\ast\,\widehat{\Theta}^\ast_\pm 
= -\int_{\bbR^{\leq 0}}\Theta_\pm^\ast \,\iota^\ast
= G^{\uparrow}_{\Theta_\pm[1]} \, \iota^\ast\quad,
\end{flalign}
where in the second step we used \cite[Section 7.12, Proposition VIII]{GHV} 
and in the third step we used the commutative diagram \eqref{eqn:bulkflowcompatibility}.
\end{proof}


\section{\label{sec:AQFTonM}The AQFT for the bulk/boundary system}
For $M$ an oriented $3$-manifold whose boundary $\partial M$ is endowed with a Lorentzian metric and a time-orientation,
we have defined in \eqref{eqn:boundarycondition} a boundary condition $\LLL^\pm$ for linear Chern-Simons theory on $M$.
This boundary condition is associated to a choice of self-dual or anti-self-dual forms on $\partial M$,
and in turn to a choice of $+$- or $-$-chiral flow $\Theta_\pm$.
In Assumption \ref{assu:bulkflows}, we have assumed the existence, and made a choice,
of a proper $\mathbb{R}$-action $\widehat{\Theta}_\pm : \mathbb{R} \times M \to M$
on the whole of $M$, which restricts to the chosen chiral flow $\Theta_\pm$ on the boundary.
In the preceding section, we demonstrated the existence of Green's homotopies 
associated to the chosen $\mathbb{R}$-action $\widehat{\Theta}_\pm$.
\sk

The goal of this section is to describe a suitable type of AQFT on $M$ that 
quantizes the boundary conditioned Chern-Simons fields
$\FFF_{\LLL^\pm}(M)$ from \eqref{eqn:conditionedfieldspullback} 
for the chosen (anti-)self-dual boundary condition $\LLL^\pm$.

\paragraph{The relevant orthogonal category:} Recall that an AQFT
can be defined in the very broad context where an \textit{orthogonal category} \cite{BSWoperad,BSWhomotopy}
is used to model a category of ``spacetime regions'' with a notion of ``independence'' of such.
The prime example for relativistic AQFT is the category $\mathbf{COpen}(N)$
of all non-empty causally convex open subsets $U\subseteq N$ of an
oriented and time-oriented globally hyperbolic Lorentzian manifold $N$,
with orthogonality relation determined by causal disjointness.
However, this choice is not suitable for our boundary conditioned Chern-Simons theory,
because the bulk $M$ does not come equipped with a Lorentzian metric.
Using our choice of proper $\bbR$-action $\widehat{\Theta}_\pm : \bbR\times M\to M$ 
that extends a chiral flow on the boundary $\partial M$, we propose the following
more suitable category.
\begin{defi}\label{def:orthogonalcategory}
We denote by $\widehat{\Theta}_\pm\text{-}\mathbf{Open}(M)$ the category
whose objects are all non-empty open subsets $U\subseteq M$ that are \textit{$\widehat{\Theta}_\pm$-convex},
in the sense that 
\begin{flalign}
J^{\uparrow}_{\widehat{\Theta}_\pm}(U)\cap J^{\downarrow}_{\widehat{\Theta}_\pm}(U)\subseteq U\quad,
\end{flalign}
and whose morphisms $\iota_U^{U^\prime} : U\to U^\prime$ 
are subset inclusions $U\subseteq U^\prime\subseteq M$. We endow this category with the orthogonality relation 
defined by $(U_1\subseteq U)\perp (U_2\subseteq U)$ if and only if $U_1$ and $U_2$ are 
\textit{$\widehat{\Theta}_\pm$-disjoint}, in the sense that 
\begin{flalign}
\big(J^\uparrow_{\widehat{\Theta}_\pm}(U_1)\cup J^\downarrow_{\widehat{\Theta}_\pm}(U_1)\big)\cap U_2 = \emptyset\quad.
\end{flalign}
\end{defi}

\begin{rem}\label{rem:alternativeconvexitydisjointness}
The concepts of $\widehat{\Theta}_\pm$-convexity and $\widehat{\Theta}_\pm$-disjointness
are analogues of causal convexity and causal disjointness in relativistic AQFT.
It is sometimes useful to rephrase these definitions in the 
following equivalent ways. A subset $U\subseteq M$ is $\widehat{\Theta}_\pm$-convex
if and only if the orbit of each point $p\in U$ under the $\bbR$-action 
$\widehat{\Theta}_\pm : \bbR\times M\to M$ does not exit and re-enter $U\subseteq M$.
Two subsets $(U_1\subseteq U)\perp (U_2\subseteq U)$ are $\widehat{\Theta}_\pm$-disjoint
if and only if the orbit of $U_1$ under the $\bbR$-action $\widehat{\Theta}_\pm : \bbR\times M\to M$
does not intersect $U_2$.
\end{rem}

The concept of Cauchy morphisms from relativistic AQFT also finds an analogue
in our context. Recall that one of the many equivalent characterizations
for a morphism $\iota_U^{U^\prime} : U\to U^\prime$ between two causally convex subsets
$U\subseteq U^\prime\subseteq N$ in an oriented and time-oriented globally hyperbolic 
Lorentzian manifold $N$ to be Cauchy is to require that the Cauchy developments 
$D(U)=D(U^\prime) \subseteq N$ in $N$ agree. This easily generalizes to our framework.
\begin{defi} \label{def:Cauchy_morphism}
A morphism $\iota_U^{U^\prime} : U\to U^\prime$ in $\widehat{\Theta}_\pm\text{-}\mathbf{Open}(M)$
is called \textit{$\widehat{\Theta}_\pm$-Cauchy} if
\begin{flalign} \label{eqn:Cauchy_morphism}
\widehat{\pi}_\pm(U) \, =\,
\widehat{\pi}_\pm (U^\prime) \quad,
\end{flalign}
where $\widehat{\pi}_\pm : M \to M /\! \sim_\pm$ is the quotient of $M$ by $\widehat{\Theta}_\pm$.
\end{defi}

\begin{rem} \label{rem:Cauchy_morphisms_and_sections}
Observe that \eqref{eqn:Cauchy_morphism} says 
equivalently that $\widehat{\pi}_\pm^{-1} \left( \widehat{\pi}_\pm(U) \right) = 
\widehat{\pi}_\pm^{-1} \left( \widehat{\pi}_\pm (U^\prime) \right)$.
Viewing the orbit $\widehat{\pi}_\pm^{-1} \left( \widehat{\pi}_\pm(U) \right)$
of $U$ under the $\bbR$-action $\widehat{\Theta}_\pm$ 
as the analogue in our context of the Cauchy development of $U$, Definition 
\ref{def:Cauchy_morphism} says that $\iota_U^{U^\prime}$ is $\widehat{\Theta}_\pm$-Cauchy 
if the (chiral analogues of) Cauchy developments of $U$ and $U^\prime$ coincide.
\sk

Recalling also Remark \ref{rem:proper-flow-Cauchy-slice},
we may regard smooth global sections of the quotient map $\widehat{\pi}_\pm : M \to M /\!\sim_\pm$
(which is a smooth principal $\mathbb{R}$-bundle by Corollary \ref{cor:proper_flow_gives_principal_bundle})
as appropriate analogues of Cauchy surfaces in our chiral geometric context.
Then $\iota_U^{U^\prime} : U \to U^\prime$ is a $\widehat{\Theta}_\pm$-Cauchy morphism
if and only if $U$ contains (the image of) a smooth section of the 
restriction $\widehat{\pi}_\pm\vert_{U^\prime} : \widehat{\pi}_\pm^{-1}(\widehat{\pi}_\pm(U^\prime)) 
\to \widehat{\pi}_\pm(U^\prime)$.
Informally, this says that a morphism is $\widehat{\Theta}_\pm$-Cauchy
if its source contains (the chiral analogue of) a Cauchy surface for its target.
\end{rem}

An AQFT on this orthogonal category is then defined as a functor
\begin{flalign}
\AAA \,:\, \widehat{\Theta}_\pm\text{-}\mathbf{Open}(M)~\longrightarrow~\dgastAlg
\end{flalign}
to the category of associative and unital differential graded $\ast$-algebras (in short, dg-algebras) 
over $\bbC$ that satisfies the Einstein causality axiom for $\perp$: 
For each orthogonal pair $(U_1\subseteq U)\perp (U_2\subseteq U)$ in $\widehat{\Theta}_\pm\text{-}\mathbf{Open}(M)$, 
the diagram
\begin{flalign}\label{eqn:Einsteincausality}
\begin{gathered}
\xymatrix@C=5em{
\ar[d]_-{\AAA(\iota_{U_1}^U)\otimes\AAA(\iota_{U_2}^U)} \AAA(U_1)\otimes\AAA(U_2)\ar[r]^-{\AAA(\iota_{U_1}^U)\otimes\AAA(\iota_{U_2}^U)} ~&~ \AAA(U)\otimes \AAA(U)\ar[d]^-{\mu_U^\op}\\
\AAA(U)\otimes \AAA(U) \ar[r]_-{\mu_U^{}}~&~ \AAA(U)
}
\end{gathered}
\end{flalign}
of cochain maps commutes, where $\mu_{U}^{(\op)} : \AAA(U)\otimes \AAA(U)\to\AAA(U)$
denotes the (opposite) multiplication on the dg-algebra $\AAA(U)\in\dgastAlg$.
Such AQFT is said to satisfy the time-slice axiom if, for each
$\widehat{\Theta}_\pm$-Cauchy morphism $\iota_U^{U^\prime}: U\to U^\prime$ 
in $\widehat{\Theta}_\pm\text{-}\mathbf{Open}(M)$,
the cochain map underlying the $\dgastAlg$-morphism 
\begin{flalign}\label{eqn:timeslice}
\AAA(\iota_U^{U^\prime}) \,:\, \AAA(U)\,\stackrel{\sim}{\longrightarrow}\, \AAA(U^\prime)
\end{flalign}
is a quasi-isomorphism.

\paragraph{Poisson cochain complexes:} Using the Green's homotopies from 
Proposition \ref{prop:G_Theta:restricts_to_FL}, we define the $(-1)$-cocycle
\begin{subequations}
\begin{flalign}
G_{\widehat{\Theta}_\pm[1]} \,:=\, G_{\widehat{\Theta}_\pm[1]}^{\uparrow} -  G_{\widehat{\Theta}_\pm[1]}^{\downarrow}\,\in\,
\big[\FFF_{\LLL^\pm,\cc}(M),\FFF_{\LLL^\pm}(M)\big]^{-1}~~,\quad \partial G_{\widehat{\Theta}_\pm[1]} \,=\,0\quad,
\end{flalign}
which equivalently can be regarded as a cochain map
\begin{flalign}
G_{\widehat{\Theta}_\pm[1]} \,:\,\FFF_{\LLL^\pm,\cc}(M)[1]~\longrightarrow~\FFF_{\LLL^\pm}(M)~~,\quad
\varphi\,\longmapsto\,G_{\widehat{\Theta}_\pm[1]}(\varphi) \,=\,-\int_{\bbR}\Theta_{\pm}^{\ast}(\varphi) 
\end{flalign}
\end{subequations}
from the $1$-shift of the complex of compactly supported sections. (Recall that 
these are the linear observables from \eqref{eqn:linobs}.) 
The role played by $G_{\widehat{\Theta}_\pm[1]}$ is analogous to the causal propagator 
(or retarded-minus-advanced Green's operator) in relativistic AQFT. In particular, we can define,
for each object $U\in \widehat{\Theta}_\pm\text{-}\mathbf{Open}(M)$, a cochain map
\begin{flalign}\label{eqn:unshiftedpoisson:preversion}
\tau_{(0)}^{U}\,:\, \xymatrix@C=1.3em{
\ar[rr]^-{{\mathrm{ext}_U^M}^{\otimes 2}} \FFF_{\LLL^\pm,\cc}(U)[1]^{\otimes 2} 
~&&~ \ar[rr]^-{\id \otimes G_{\widehat{\Theta}_\pm[1]}} \FFF_{\LLL^\pm,\cc}(M)[1]^{\otimes 2} ~&&~\FFF_{\LLL^\pm,\cc}(M)[1]\otimes \FFF_{\LLL^\pm}(M) \ar[r]^-{\ev}~&~\bbR
}\quad,
\end{flalign}
where $\mathrm{ext}$ denotes extension by zero of compactly supported sections
and $\ev$ was defined in \eqref{eqn:ev} in terms of integration over $M$.
Note that, by construction, this family of cochain maps is natural on the category 
$ \widehat{\Theta}_\pm\text{-}\mathbf{Open}(M)$ in the sense that, 
given any morphism $U\subseteq U^\prime$ in this category, we have that
\begin{flalign}\label{eqn:taunaturality}
\tau_{(0)}^{U^\prime}\circ \big(\mathrm{ext}_{U}^{U^\prime} \otimes \mathrm{ext}_{U}^{U^\prime}\big) \,=\,\tau_{(0)}^{U}\quad.
\end{flalign}
To ease notation, we shall often drop the superscript ${}^U$ and simply write $\tau_{(0)}$
for every member of the family $\{\tau_{(0)}^U\}$ of cochain maps.
\sk

To show that \eqref{eqn:unshiftedpoisson:preversion} defines indeed an 
unshifted linear Poisson structure, i.e.\ $\tau_{(0)}$ is graded antisymmetric, 
we make use of the following result that allows us to rewrite
this map in terms of an integral over the base manifold $B_\pm$ of the quotient map 
$\widehat{\pi}_\pm : M\to M/\!\sim_\pm \,=: B_\pm$.
Recall that the latter is a smooth principal $\bbR$-bundle by 
Corollary \ref{cor:proper_flow_gives_principal_bundle},
and that this induces on $B_\pm$ an orientation that is
appropriately compatible with the orientation of $M$, see
Corollary \ref{cor:principal_R-bundles:total_base_orientations}.
\begin{lem} \label{lem:unshifted_Poisson_on_base}
The cochain map \eqref{eqn:unshiftedpoisson:preversion} can be written as
\begin{flalign}
\tau_{(0)}(\varphi\otimes \psi) \,=\, (-1)^{\vert \varphi\vert } \,\int_{B_\pm} (\widehat{\pi}_{\pm})_\ast(\varphi)\wedge(\widehat{\pi}_{\pm})_\ast(\psi)\quad,
\end{flalign}
for all homogeneous $\varphi,\psi\in  \FFF_{\LLL^\pm,\cc}(U)[1]$.
Here, $(\widehat{\pi}_\pm)_\ast :\FFF_{\LLL^\pm,\cc}(M)[1]\subseteq\Omega_\cc^\bullet(M)[2]\to
\Omega^\bullet_\cc(B_\pm)[1] $ denotes fiber integration along the quotient map
$\widehat{\pi}_\pm : M\to M/\!\sim_\pm \,=: B_\pm$.
\end{lem}
\begin{proof}
From Proposition \ref{prop:G_Theta:difference},
we have $G_{\widehat{\Theta}_\pm[1]} = -G_{\widehat{\Theta}_\pm} = - (\widehat{\pi}_\pm)^\ast (\widehat{\pi}_\pm)_\ast$.
Then
\begin{flalign}
\tau_{(0)}(\varphi \otimes \psi) & =
(-1)^{\vert \varphi\vert +1} \,\int_M \varphi \wedge G_{\widehat{\Theta}_\pm[1]} (\psi)
= (-1)^{\vert\varphi\vert} \,\int_M \varphi \wedge (\widehat{\pi}_\pm)^\ast (\widehat{\pi}_\pm)_\ast (\psi) \
\nonumber \\
& = (-1)^{\vert\varphi\vert} \, \int_{B_\pm} (\widehat{\pi}_\pm)_\ast (\varphi) \wedge (\widehat{\pi}_\pm)_\ast (\psi)\quad,
\end{flalign}
for all homogeneous $\varphi,\psi\in  \FFF_{\LLL^\pm,\cc}(U)[1]$.
The last equality uses the Fubini theorem for fiber integration, see e.g.\ \cite[Proposition 3.4.48]{Nicolaescu}.
\end{proof}

From this equivalent description, it is immediate
to see that the cochain maps \eqref{eqn:unshiftedpoisson:preversion} are graded antisymmetric, 
hence they define unshifted linear Poisson structures. Using
the naturality observed in \eqref{eqn:taunaturality},
we further obtain that the assignment
$U\mapsto \big(\FFF_{\LLL^\pm,\cc}(U)[1],\tau_{(0)}^{U}\big)$
extends to a functor
\begin{flalign}\label{eqn:Poissonfunctor}
\big(\FFF_{\LLL^\pm,\cc}[1],\tau_{(0)}\big) \,:\,  \widehat{\Theta}_\pm\text{-}\mathbf{Open}(M) \,\longrightarrow\, \mathbf{PoCh}_\bbR
\end{flalign}
taking values in the category of Poisson cochain complexes.
\begin{propo}\label{prop:propertiesPoissonfunctor}
The functor \eqref{eqn:Poissonfunctor} satisfies the following classical analogues 
of the Einstein causality and time-slice axioms:
\begin{itemize}
\item[(a)] For each orthogonal pair $(U_1\subseteq U)\perp (U_2\subseteq U)$ 
in $\widehat{\Theta}_\pm\text{-}\mathbf{Open}(M)$,
\begin{flalign}
\tau^U_{(0)}\circ\big(\mathrm{ext}_{U_1}^{U}\otimes\mathrm{ext}_{U_2}^U\big)\,=\,0
\end{flalign}
vanishes.

\item[(b)] For each $\widehat{\Theta}_\pm$-Cauchy morphism $\iota_U^{U^\prime}: U\to U^\prime$ 
in $\widehat{\Theta}_\pm\text{-}\mathbf{Open}(M)$, the cochain map
\begin{flalign}
\mathrm{ext}_U^{U^\prime}\,:\, \FFF_{\LLL^\pm,\cc}(U)[1]\,\stackrel{\sim}{\longrightarrow}\,\FFF_{\LLL^\pm,\cc}(U^\prime)[1]
\end{flalign}
is a quasi-isomorphism.
\end{itemize}
\end{propo}
\begin{proof}
Item (a) follows immediately from the support properties of Green's homotopies
in Definition \ref{def:Greenshomotopies}. (Recall that the map $\ev$ in \eqref{eqn:unshiftedpoisson:preversion} 
is given by integration over $M$, see \eqref{eqn:ev}.)
The proof of item (b) is analogous to the proof of the corresponding
statement in the relativistic case, see e.g.\ \cite[Theorem 3.10]{BMSmodels}.
Where in the relativistic case one uses Cauchy surfaces to construct a quasi-inverse,
we use sections of the quotient map $\widehat{\pi}_\pm : M \to B_\pm$, 
recalling Remarks \ref{rem:proper-flow-Cauchy-slice} and \ref{rem:Cauchy_morphisms_and_sections}.
The compactness property used in the relativistic case may be substituted by 
Lemma \ref{lem:intersect-future-compact-past-Cauchy}.
\end{proof}

\paragraph{The AQFT:} The quantization of our bulk/boundary system
is easy to carry out by using the cochain complex-valued canonical commutation relations functor
$\CCR : \mathbf{PoCh}_\bbR\to \dgastAlg$ from \cite[Section 5]{LinearYM}.
Post-composing the functor \eqref{eqn:Poissonfunctor} with $\CCR$ 
defines the $\dgastAlg$-valued functor
\begin{flalign}\label{eqn:AAApm}
\AAA_\pm \,:=\,\CCR\big(\FFF_{\LLL^\pm,\cc}[1],\tau_{(0)}\big)\,:\, 
\widehat{\Theta}_\pm\text{-}\mathbf{Open}(M) \,\longrightarrow\, \dgastAlg
\end{flalign}
that assigns to $U\in \widehat{\Theta}_\pm\text{-}\mathbf{Open}(M)$ the dg-algebra
\begin{subequations}
\begin{flalign}
\AAA_\pm(U) \,=\, \CCR\big(\FFF_{\LLL^\pm,\cc}(U)[1],\tau_{(0)}^U\big) \,=\,T^\otimes_\bbC \left(\FFF_{\LLL^\pm,\cc}(U)[1]\right)\big/\,\mathcal{I}_{\tau^{U}_{(0)}}\,\in\,\dgastAlg
\end{flalign}
constructed as the complexified tensor algebra over the cochain
complex $\FFF_{\LLL^\pm,\cc}(U)[1]$ of linear observables,
modulo the two-sided differential graded $\ast$-ideal $\mathcal{I}_{\tau^U_{(0)}}$ that is
generated by the canonical commutation relations
\begin{flalign}
\varphi\otimes \psi - (-1)^{\vert \varphi\vert\,\vert\psi\vert}\,\psi\otimes \varphi = \ii\,\tau_{(0)}^U\big(\varphi\otimes \psi\big)\,\oone\quad,
\end{flalign}
\end{subequations}
for all homogeneous generators $\varphi,\psi\in \FFF_{\LLL^\pm,\cc}(U)[1]$. The functorial
structure is given by extending to dg-algebra morphisms 
the extension-by-zero maps $\mathrm{ext}$. As a consequence of Proposition
\ref{prop:propertiesPoissonfunctor} and standard arguments, see e.g.\ \cite[Lemma 6.18]{LinearYM},
it follows that \eqref{eqn:AAApm} satisfies the AQFT axioms. Let us record this result.
\begin{cor}\label{cor:AQFTbulk}
The functor $\AAA_\pm : \widehat{\Theta}_\pm\text{-}\mathbf{Open}(M) \to \dgastAlg$ 
given in \eqref{eqn:AAApm} satisfies the Einstein causality axiom \eqref{eqn:Einsteincausality}
and the time-slice axiom \eqref{eqn:timeslice}.
\end{cor}


\section{\label{sec:dimred}Dimensional reduction to the base manifold}
As a consequence of the weak variant of the time-slice axiom \eqref{eqn:timeslice}
in terms of \textit{quasi-isomorphisms}, which is sometimes called the homotopy time-slice axiom,
it is difficult to analyze the AQFT $\AAA_\pm$ constructed in Corollary \ref{cor:AQFTbulk} directly.
The goal of this section is to replace $\mathfrak{A}_\pm$ by a weakly 
equivalent model $\mathfrak{B}_\pm$ living on the base space of the quotient
map $\widehat{\pi}_\pm : M\to M/\!\sim_\pm =:B_\pm$ associated with the proper $\bbR$-action
$\widehat{\Theta}_\pm : \bbR\times M\to M$.
This base space $B_\pm$ is an oriented, smooth $2$-manifold, see 
Corollaries \ref{cor:proper_flow_gives_principal_bundle} and \ref{cor:principal_R-bundles:total_base_orientations},
so one achieves a dimensional reduction from
the $3$-dimensional manifold $M$ to the $2$-dimensional manifold $B_\pm$.
Note that the base manifold $B_\pm:=M/\!\sim_\pm$ has a non-empty boundary $\partial B_\pm = \partial M/\!\sim_\pm$
since $\partial M$ is a $\widehat{\Theta}_\pm$-invariant subspace of $M$ as per \eqref{eqn:G-bulk-boundary-compatibility}.
Hence, our dimensional reduction allows us to analyze boundary features
of our bulk/boundary system by working in a simpler but equivalent $2$-dimensional setup.
\sk

It is important to stress that this dimensional reduction is \textit{not} to the boundary $\partial M$ of $M$,
and is therefore not itself exhibiting a CS/WZW correspondence. (The latter will be studied
in Section \ref{sec:analysis} below.)
The present dimensional reduction should be compared to the description of chiral CFTs by conformal nets.
In that case, a chiral conformal theory on a $2$-dimensional spacetime
(say Minkowski $\mathbb{R}^{1+1}$ or an appropriate compactification thereof)
is equivalently described by data on a $1$-dimensional light ray in that spacetime.
This is possible because the fields in a chiral theory are independent of the 
position along the transverse light ray.
The analogous independence exploited by our dimensional reduction
is exactly that expressed by the time-slice axiom \eqref{eqn:timeslice}.
\sk

To construct this dimensionally reduced model, we begin by replacing $\mathfrak{A}_\pm$ by 
a weakly equivalent model that satisfies
the time-slice axiom strictly in terms of \textit{isomorphisms}. The relevant mathematical
framework for such strictification constructions has been developed in \cite{BCStimeslice},
in the context of model category theory.

\paragraph{Orthogonal localization:}
We describe an explicit model for the localization of the orthogonal category 
$\widehat{\Theta}_\pm\text{-}\mathbf{Open}(M)$ from Section \ref{sec:AQFTonM} at all 
$\widehat{\Theta}_\pm$-Cauchy morphisms. We also will show that this localization is reflective,
hence \cite[Theorem 3.6]{BCStimeslice} provides the desired strictification construction.
Let us start by defining the orthogonal categories that we will identify later as the desired localizations.
\begin{defi}
We denote by $\mathbf{Open}(B_\pm)$ the category whose objects
are all non-empty open subsets $V\subseteq B_\pm$ of the base manifold of the quotient map 
$\widehat{\pi}_\pm : M\to M/\!\sim_\pm =:B_\pm$ and whose morphisms $\iota_V^{V^\prime} : V\to V^\prime$
are subset inclusions $V\subseteq V^\prime\subseteq B_\pm$. We endow this category with
the disjointness orthogonality relation, i.e.\ $(V_1\subseteq V)\perp(V_2\subseteq V)$ if and only
if $V_1\cap V_2 =\emptyset$.
\end{defi}

Taking preimages under $\widehat{\pi}_\pm : M\to B_\pm$ defines a functor
\begin{flalign}\label{eqn:localizationreflector}
\widehat{\pi}_\pm^{-1}\,:\, \mathbf{Open}(B_\pm)\,\longrightarrow \, \widehat{\Theta}_\pm\text{-}\mathbf{Open}(M)
\end{flalign}
to the category introduced in Definition \ref{def:orthogonalcategory} since
the preimage $\widehat{\pi}_\pm^{-1}(V)\subseteq M$ of any non-empty open subset $V\subseteq B_\pm$
is non-empty, open and $\widehat{\Theta}_\pm$-convex. This functor preserves
the orthogonality relations on the respective categories, i.e.\ it defines an orthogonal functor
in the sense of \cite{BSWoperad,BCStimeslice}.
Since the fiber bundle projection $\widehat{\pi}_\pm : M\to B_\pm$ is an open map,
we can also take images and define another functor
\begin{flalign}\label{eqn:localizationfunctor}
\widehat{\pi}_\pm\,:\,\widehat{\Theta}_\pm\text{-}\mathbf{Open}(M)\,\longrightarrow\, \mathbf{Open}(B_\pm)
\end{flalign}
that goes in the opposite direction. Also this functor is easily seen to preserve
the orthogonality relations. We now observe that these two (orthogonal) functors 
define an adjunction $\widehat{\pi}_\pm \dashv \widehat{\pi}_\pm^{-1}$ with counit 
$\epsilon :\widehat{\pi}_\pm \widehat{\pi}_\pm^{-1} \Rightarrow \id $ defined
component-wise by the identity morphisms $\widehat{\pi}_\pm \widehat{\pi}_\pm^{-1}(V) = V$
and unit $\eta: \id \Rightarrow \widehat{\pi}_\pm^{-1}\,\widehat{\pi}_\pm$ defined 
component-wise by the inclusions $U\subseteq \widehat{\pi}_\pm^{-1}(\widehat{\pi}_\pm(U))$.
\begin{propo}
The functor \eqref{eqn:localizationfunctor} defines an orthogonal
localization of $\widehat{\Theta}_\pm\text{-}\mathbf{Open}(M)$ at the set of
all $\widehat{\Theta}_\pm$-Cauchy morphisms. This orthogonal localization
is reflective with reflector \eqref{eqn:localizationreflector}
\end{propo}
\begin{proof}
A morphism $\iota_U^{U^\prime} : U\to U^\prime$ in $\widehat{\Theta}_\pm\text{-}\mathbf{Open}(M)$
is $\widehat{\Theta}_\pm$-Cauchy according to Definition \ref{def:Cauchy_morphism} exactly 
when the functor $\widehat{\pi}_\pm$ sends it to an identity morphism.
Since the right adjoint functor \eqref{eqn:localizationreflector} is fully faithful,
the statement then follows from \cite[Proposition 3.3]{BCStimeslice}
and the fact that the only isomorphisms in $\mathbf{Open}(B_\pm)$ are identities.
\end{proof}

From \cite[Theorem 3.6]{BCStimeslice}, we obtain an
AQFT $\AAA_\pm^\mathrm{st} : \widehat{\Theta}_\pm\text{-}\mathbf{Open}(M)\to\dgastAlg$
that is weakly equivalent to the AQFT $\AAA_\pm$ from Corollary \ref{cor:AQFTbulk} but has the technical advantage
that it satisfies the strict time-slice axiom, i.e.\ to any $\widehat{\Theta}_\pm$-Cauchy morphism
$\iota_U^{U^\prime} : U\to U^\prime$ it assigns an \textit{isomorphism} $\AAA_\pm^{\mathrm{st}}(\iota_U^{U^\prime}):
\AAA_\pm^{\mathrm{st}}(U)\stackrel{\cong}{\longrightarrow} \AAA_\pm^{\mathrm{st}}(U^\prime)$ in contrast to a
quasi-isomorphism. Using also \cite[Corollary 3.7]{BCStimeslice},
one can construct this AQFT very explicitly as the pullback 
$\AAA_\pm^{\mathrm{st}} = (\widehat{\pi}_\pm)^\ast (\widehat{\pi}_\pm^{-1})^\ast(\AAA_\pm)$ 
of the original $\AAA_\pm$ along the two orthogonal functors \eqref{eqn:localizationreflector} 
and \eqref{eqn:localizationfunctor}. 
Observe that the intermediate AQFT $(\widehat{\pi}_\pm^{-1})^\ast(\AAA_\pm) :  \mathbf{Open}(B_\pm) \to \dgastAlg$
in this construction is defined on the base manifold $B_\pm$. By \cite[Theorem 3.6]{BCStimeslice}, this
AQFT carries equivalent information to $\AAA_\pm$, and hence also to its strictification $\AAA^{\mathrm{st}}_\pm$, 
which means that we have achieved a dimensional reduction from the $3$-manifold $M$ to the 
$2$-dimensional base manifold $B_\pm$.

\paragraph{Simplification:} Even though the AQFT
$(\widehat{\pi}_\pm^{-1})^\ast(\AAA_\pm) :  \mathbf{Open}(B_\pm) \to \dgastAlg$
is defined on the category of non-empty opens of the $2$-dimensional manifold $B_\pm$, 
the dg-algebras
\begin{flalign}
\AAA_\pm\big(\widehat{\pi}_\pm^{-1}(V)\big) \,=\,  
\CCR\Big(\FFF_{\LLL^\pm,\cc}\big(\widehat{\pi}_\pm^{-1}(V)\big)[1],\tau_{(0)}\Big)
\end{flalign}
it assigns to objects $(V\subseteq B_\pm)\in  \mathbf{Open}(B_\pm)$ are determined from
$3$-dimensional data on $(\widehat{\pi}_\pm^{-1}(V) \subseteq M)\in \widehat{\Theta}_\pm\text{-}\mathbf{Open}(M)$.
The aim of this paragraph is to construct a weakly equivalent model for the AQFT $(\widehat{\pi}_\pm^{-1})^\ast(\AAA_\pm)$
that is determined only from $2$-dimensional data on $B_\pm$. Inspired by Lemma \ref{lem:unshifted_Poisson_on_base}, 
we consider the fiber integration
\begin{flalign}\label{eqn:fiberintegrationsimplification}
(\widehat{\pi}_\pm)_\ast \,:\,\FFF_{\LLL^\pm,\cc}\big(\widehat{\pi}_\pm^{-1}(V)\big)[1]\subseteq
 \Omega^\bullet_\cc\big(\widehat{\pi}_\pm^{-1}(V)\big)[2]\,\longrightarrow\,
\Omega^\bullet_\cc(V)[1] 
\end{flalign}
along the quotient map $\widehat{\pi}_\pm :\widehat{\pi}_\pm^{-1}(V)\subseteq M \to V\subseteq B_\pm$.
\begin{lem}\label{lem:dimregqiso}
The cochain map \eqref{eqn:fiberintegrationsimplification} factors through the subcomplex
\begin{flalign}\label{eqn:Blinearobservables}
\Omega^\bullet_{\partial,\cc} (V)[1]\,:=\, \Big(
\xymatrix@C=2em{
\stackrel{(-1)}{\Omega_{\partial,\cc}^0(V)} \ar[r]^-{-\dd}
& \stackrel{(0)}{\Omega^1_{\cc}(V)}\ar[r]^-{-\dd} 
& \stackrel{(1)}{\Omega^2_\cc(V)}
}
\Big) \,\subseteq\, \Omega^\bullet_\cc(V)[1] \quad,
\end{flalign}
where $\Omega_{\partial,\cc}^0(V)\subseteq \Omega_\cc^0(V)$ denotes the
subspace of all compactly supported $0$-forms on $V$ whose pullback to the boundary $\partial V$ vanishes. (In the case
of $\partial V=\emptyset$, this boundary condition becomes void.)
Moreover, the corestriction
\begin{flalign}\label{eqn:fiberintegrationcorestriction}
(\widehat{\pi}_\pm)_\ast \,:\,\FFF_{\LLL^\pm,\cc}\big(\widehat{\pi}_\pm^{-1}(V)\big)[1]\,\stackrel{\sim}{\longrightarrow}\,
\Omega^\bullet_{\partial,\cc}(V)[1] 
\end{flalign}
is a quasi-isomorphism
for each $(V\subseteq B_\pm)\in  \mathbf{Open}(B_\pm)$.
\end{lem}
\begin{proof}
To simplify this proof, we pick a trivialization $\widehat{\pi}_\pm^{-1}(V) \cong \bbR\times V$ 
of the principal $\bbR$-bundle $\widehat{\pi}_\pm :\widehat{\pi}_\pm^{-1}(V)\subseteq M \to V\subseteq B_\pm$ 
and denote by $\tau$ a choice of global coordinate on $\bbR$.
The vector space of (compactly supported) $k$-forms on 
the product manifold $\bbR\times V$ admits a decomposition $\Omega^k_{\cc}(\bbR\times V) = 
\Omega^{1,k-1}_\cc(\bbR\times V) \oplus \Omega^{0,k}_\cc(\bbR\times V)$ into forms that do 
have a leg along $\bbR$ and forms that do not have a leg along $\bbR$.
(See also \eqref{eqn:forms_split_coordinates} and the discussion preceding it.)
Note that
forms of the first type look like $\dd \tau \wedge \phi$, with $\phi\in \Omega^{k-1}_\cc(\bbR\times V)$
a $(k{-}1)$-form with trivial contraction $\iota_{\partial_\tau}\phi =0$ along the vector field $\partial_\tau$.
Forms of the second type are $\psi\in \Omega^{k}_\cc(\bbR\times V)$ satisfying $\iota_{\partial_\tau}\psi =0$.
Recall that the boundary condition $\mathfrak{L}^\pm$ determining 
$\FFF_{\LLL^\pm,\cc}(\bbR\times V)[1]\subseteq
\Omega^\bullet_\cc(\bbR\times V)[2]$ only affects $0$-forms and $1$-forms, cf.\ \eqref{eqn:boundarycondition}.
Concretely, we have that a $0$-form $\psi\in \Omega^{0}_\cc(\bbR\times V)$ satisfies 
the boundary condition if and only if its restriction to $\bbR\times \partial V$ vanishes,
a $1$-form of the first type $\dd \tau \wedge \phi\in  \Omega^{1}_\cc(\bbR\times V)$ satisfies the boundary condition
if and only if the restriction of $\phi\in \Omega^{0}_\cc(\bbR\times V)$ to $\bbR\times \partial V$ vanishes, 
and every $1$-form of the second type $\psi\in \Omega^{1}_\cc(\bbR\times V)$ satisfies the boundary condition
as a consequence of Lemma \ref{lem:sd/asdannihilation}.
\sk

To prove the statement about the factorization, we have to show that,
for every $1$-form $\varphi\in \FFF_{\LLL^\pm,\cc}(\bbR\times V)[1]^{-1}\subseteq \Omega^{1}_\cc(\bbR\times V)$
satisfying the boundary condition, the restriction of 
$(\widehat{\pi}_\pm)_\ast(\varphi) \in \Omega^0_\cc(V)$ to the boundary $\partial V$
vanishes. For $1$-forms of the second type this is trivial since $(\widehat{\pi}_\pm)_\ast(\psi) =0$
and for $1$-forms of the first type $\dd \tau \wedge \phi\in  \Omega^{1}_\cc(\bbR\times V)$
this follows from the fact that $\phi\in \Omega^{0}_\cc(\bbR\times V)$ vanishes on $\bbR\times \partial V$
as a consequence of the boundary condition. 
\sk

To prove that \eqref{eqn:fiberintegrationcorestriction} is a quasi-isomorphism,
we construct an explicit quasi-inverse by adapting the standard proof of the compactly 
supported Poincar{\'e} lemma, see e.g.\ \cite[Proposition 4.6]{BottTu}. 
Let us pick a compactly supported $1$-form $\omega\in \Omega^1_\cc(\bbR)$ with unit integral $\int_\bbR\omega=1$
and define the cochain map
\begin{flalign}
\nn \omega_\ast \,:\, \Omega^\bullet_{\partial,\cc}(V)[1]\,&\longrightarrow\,
\FFF_{\LLL^\pm,\cc}(\bbR\times V)[1]\subseteq \Omega^\bullet_\cc\big(\bbR\times V\big)[2]\quad,\\
\alpha\,&\longmapsto\,\pr_\bbR^\ast(\omega)\wedge \pr_V^\ast(\alpha) =: \omega\wedge\alpha\quad.
\end{flalign}
To simplify notation, we will always suppress pullbacks along the projection maps
$\pr_{\bbR}$ and $\pr_{V}$ on the two factors of $\bbR\times V$. 
The composition $(\widehat{\pi}_\pm)_\ast\circ \omega_\ast =\id$
is the identity since the integral of $\omega\in \Omega^1_\cc(\bbR)$ is $1$.
The other composition is homotopic to the identity, i.e.\ $\id- \omega_\ast\circ (\widehat{\pi}_\pm)_\ast  = \partial K$,
for the cochain homotopy 
$K\in \big[\FFF_{\LLL^\pm,\cc}(\bbR\times V)[1],\FFF_{\LLL^\pm,\cc}(\bbR\times V)[1]\big]^{-1}$ 
defined by
\begin{flalign}
K(\varphi) \,:=\, \int_{-\infty}^\tau\varphi - \left(\int_{-\infty}^\tau\omega\right)\wedge (\widehat{\pi}_\pm)_\ast(\varphi)\quad,
\end{flalign}
for all homogeneous $\varphi\in \FFF_{\LLL^\pm,\cc}(\bbR\times V)[1]$,
where $\int_{-\infty}^\tau$ denote restricted fiber integrations.
It is easy to check that $K$ is compatible with the boundary conditions 
by using their explicit description from the beginning of this proof.
\end{proof}

\begin{rem}\label{rem:fieldcomment}
The quasi-isomorphism $(\widehat{\pi}_\pm)_\ast$ in 
\eqref{eqn:fiberintegrationcorestriction} describes how the
linear observables $\FFF_{\LLL^\pm,\cc}\big(\widehat{\pi}_\pm^{-1}(V)\big)[1]$
of our $3$-dimensional theory are identified with the 
linear observables $\Omega^\bullet_{\partial,\cc}(V)[1]$
of a $2$-dimensional theory. As some readers might prefer thinking in terms
of fields instead of observables, let us dualize this construction
using the evaluation pairings between linear observables and fields.
The $3$-dimensional fields are described by the boundary conditioned Chern-Simons complex 
$\FFF_{\LLL^\pm}\big(\widehat{\pi}_\pm^{-1}(V)\big)$ from
\eqref{eqn:fieldcomplex:boundarycondition} and the $2$-dimensional
fields by the boundary conditioned $1$-shifted de Rham complex $\Omega_{\partial}^\bullet(V)[1]$,
for the boundary condition demanding that the pullback of $0$-forms to the boundary $\partial V$ vanishes.
(Note that there are no support restrictions in the field complexes.)
The evaluation pairings are given by \eqref{eqn:ev} in the $3$-dimensional case 
and the analogous integration pairing in $2$ dimensions. The dual
of \eqref{eqn:fiberintegrationcorestriction} is then given by the pullback map
$(\widehat{\pi}_\pm)^\ast : \Omega_{\partial}^\bullet(V)[1] \to \FFF_{\LLL^\pm}\big(\widehat{\pi}_\pm^{-1}(V)\big)$
on field complexes. Choosing any section of the principal $\bbR$-bundle $\widehat{\pi}_\pm : \widehat{\pi}_\pm^{-1}(V)\to V$,
we obtain an analogue of a Cauchy surface as per Remark \ref{rem:proper-flow-Cauchy-slice}
and can interpret this map as the constant
evolution of initial data in $\Omega_{\partial}^\bullet(V)[1]$ along the fibers of $\widehat{\pi}_\pm^{-1}(V)$.
Hence, the origin of the quasi-isomorphism \eqref{eqn:fiberintegrationcorestriction}
lies in the fact that Chern-Simons theory is locally constant (i.e.\ topological) 
along the fibers of $\widehat{\pi}_\pm:\widehat{\pi}_\pm^{-1}(V)\to V$.
\end{rem}

Using Corollary \ref{cor:principal_R-bundles:total_base_orientations}, 
we choose the orientation of $B_\pm$ to be compatible with the orientation of $M$.
With the induced orientation on $V \subseteq B_\pm$,
the cochain complexes \eqref{eqn:Blinearobservables} can be endowed with
the following unshifted linear Poisson structures
\begin{flalign}\label{eqn:sigmapoisson}
\sigma_{(0)}^{V}\,:\, \Omega^\bullet_{\partial,\cc} (V)[1]\otimes \Omega^\bullet_{\partial,\cc} (V)[1]
\,\longrightarrow\,\bbR~,~~\alpha\otimes\beta\,\longmapsto\,(-1)^{\vert \alpha\vert}\,
\int_{V} \alpha \wedge \beta\quad,
\end{flalign}
which are natural with respect to the extension-by-zero maps 
$\ext_{V}^{V^\prime} :  \Omega^\bullet_{\partial,\cc} (V)[1]\to  \Omega^\bullet_{\partial,\cc} (V^\prime)[1]$.
This defines a functor 
\begin{flalign}\label{eqn:reduced}
\big( \Omega^\bullet_{\partial,\cc}[1],\sigma_{(0)}\big)\,:\, \mathbf{Open}(B_\pm)\,\longrightarrow\,\mathbf{PoCh}_\bbR
\end{flalign}
and we observe that fiber integration 
\eqref{eqn:fiberintegrationcorestriction} defines a natural transformation
\begin{flalign}\label{eqn:fiberintegrationcorestriction-natural}
(\widehat{\pi}_\pm)_\ast \,:\, \big(\FFF_{\LLL^\pm,\cc}[1],\tau_{(0)}\big) \circ \widehat{\pi}_\pm^{-1} 
\,\Longrightarrow\,\big( \Omega^\bullet_{\partial,\cc}[1],\sigma_{(0)}\big)
\end{flalign}
between functors from $\mathbf{Open}(B_\pm)$ to $\mathbf{PoCh}_\bbR$.
Indeed, the component cochain maps $(\widehat{\pi}_\pm)_\ast$ are  $\mathbf{PoCh}_\mathbb{R}$-morphisms by 
Lemma \ref{lem:unshifted_Poisson_on_base}. 
This allows us to obtain our desired weakly equivalent model
for the AQFT $(\widehat{\pi}_\pm^{-1})^\ast(\AAA_\pm)$ that is defined only in terms
of geometric data on the $2$-dimensional base manifold $B_\pm$. 
Explicitly, consider the functor 
\begin{flalign}\label{eqn:BBBpm}
\BBB_\pm\, :=\, \CCR\big( \Omega^\bullet_{\partial,\cc}[1],\sigma_{(0)}\big)\,:\,\mathbf{Open}(B_\pm) \,\longrightarrow\, \dgastAlg 
\end{flalign} 
defined by post-composing the functor \eqref{eqn:reduced} with 
the canonical commutation relations functor
$\CCR : \mathbf{PoCh}_\bbR\to \dgastAlg$ from \cite[Section 5]{LinearYM}.
Note that $\BBB_\pm$ satisfies the Einstein causality axiom 
on $\mathbf{Open}(B_\pm)$ because \eqref{eqn:reduced} satisfies 
a classical analogue of the Einstein causality axiom. 
Indeed, for each orthogonal (i.e.\ disjoint) pair 
$(V_1\subseteq V)\perp (V_2\subseteq V)$ in $\mathbf{Open}(B_\pm)$, 
$\sigma^V_{(0)}\circ\big(\mathrm{ext}_{V_1}^{V}\otimes\mathrm{ext}_{V_2}^V\big)=0$ 
vanishes manifestly, see \eqref{eqn:sigmapoisson}. 
\begin{cor}\label{cor:dimreducedAQFT}
The morphism (i.e.\ natural transformation)
\begin{flalign}
\CCR\big((\widehat{\pi}_\pm)_\ast\big)\,:\,(\widehat{\pi}_\pm^{-1})^\ast(\AAA_\pm)\, \stackrel{\sim}{\Longrightarrow}\, \BBB_\pm
\end{flalign}
of AQFTs on $\mathbf{Open}(B_\pm)$, defined by applying the 
$\CCR$-functor to the natural transformation \eqref{eqn:fiberintegrationcorestriction-natural}, 
is a natural weak equivalence. Hence, the AQFT $\BBB_\pm$ from \eqref{eqn:BBBpm} provides
a weakly equivalent model for $(\widehat{\pi}_\pm^{-1})^\ast(\AAA_\pm)$.
\end{cor}
\begin{proof}
The statement follows directly from Lemma \ref{lem:dimregqiso} and 
the fact that the $\CCR$-functor preserves weak equivalences, see \cite[Proposition 5.3]{LinearYM}.
\end{proof}

\paragraph{Non-dependence on the choice of chiral flow extension:} In order to 
construct the linear Chern-Simons AQFT $\mathfrak{A}_\pm$ via Green's homotopies,
we have assumed extra data on $M$
in the form of the proper $\bbR$-action $\widehat{\Theta}_\pm : \bbR\times M\to M$.
In the dimensionally reduced model $\mathfrak{B}_\pm$,
this extra data is encoded as the quotient space $B_\pm := M /\!\sim_\pm$
on which the AQFT is defined.
\sk

We now prove,
under an additional topological assumption on $M$,
that the dimensionally reduced theories $\mathfrak{B}_\pm$ and $\mathfrak{B}^\prime_\pm$
associated to two different choices of proper $\bbR$-actions 
$\widehat{\Theta}_\pm$ and $\widehat{\Theta}^\prime_\pm$ on $M$
are equivalent in an appropriate sense.
Since $\mathfrak{A}_\pm$ is fully determined by its dimensionally reduced counterpart,
it follows that our construction is independent of the particular choice of proper $\bbR$-action $\widehat{\Theta}_\pm$ on $M$.
The important content of Assumption \ref{assu:bulkflows} for our construction is therefore only:
(a) a choice of chirality $\pm$ on $\partial M$, and
(b) the \emph{existence} of a proper $\bbR$-action on $M$
which restricts on $\partial M$ to a chiral flow of the chosen chirality.
\sk

Recall that the quotient space $B_\pm$ of $M$ by the proper $\bbR$-action $\widehat{\Theta}_\pm$
is a smooth $2$-manifold with orientation suitably compatible with the orientation of $M$,
see Corollaries \ref{cor:proper_flow_gives_principal_bundle} and \ref{cor:principal_R-bundles:total_base_orientations}.
\begin{propo} \label{prop:diffeomorphic-bases}
Suppose that $\widehat{\Theta}_\pm$ and $\widehat{\Theta}^\prime_\pm$ 
are two proper $\bbR$-actions on the oriented $3$-manifold $M$
which restrict to the boundary $\partial M$ as per \eqref{eqn:bulkflowcompatibility}
to give chiral flows $\Theta_\pm$ and $\Theta^\prime_\pm$ of the same chirality.
If each component of $M$ has more than one distinct topological end,
then there exists an orientation-preserving diffeomorphism $f : B_\pm \to B^\prime_\pm$
between the	quotient spaces of $M$ by $\widehat{\Theta}_\pm$ and $\widehat{\Theta}^\prime_\pm$.
\end{propo}
\begin{proof}
We describe the case where $M$ is connected,
hence $B_\pm$ and $B^\prime_\pm$ are also connected.
The result extends component-wise for disconnected $M$.
Since any principal $\mathbb{R}$-bundle is trivializable,
there exist diffeomorphisms $M \cong \mathbb{R} \times B_\pm \cong \mathbb{R} \times B^\prime_\pm$.
The hypothesis that $M$ has more than one end implies that that $B_\pm$ is compact.
Indeed, if $B_\pm$ were non-compact then the product $M \cong \mathbb{R} \times B_\pm$
of connected, locally connected, non-compact spaces would have exactly one end, see \cite[Theorem 5.2]{ILM}.
Similarly, $B^\prime_\pm$ is also compact.
\sk

Since $\mathbb{R}$ deformation retracts to a point, it follows
that $B_\pm$ and $B_\pm^\prime$ are homotopy equivalent,
and hence have the same genus. Since the boundary restrictions
$\Theta_\pm$ and $\Theta^\prime_\pm$ of the $\bbR$-actions have the same chirality, 
it follows that $\partial B_\pm = \partial M /\! \sim_\pm \,= \partial M /\! \sim_\pm^\prime \,= \partial B_\pm^\prime$,
see Remark \ref{rem:canonical_flow}. Then by the classification of orientable compact surfaces
(see Remark \ref{rem:classification_of_surfaces} below) there exists an 
orientation-preserving diffeomorphism $f : B_\pm \to B^\prime_\pm$.
\end{proof}

\begin{rem} \label{rem:classification_of_surfaces}
The preceding proof exploits the classification of orientable compact surfaces.
This may be stated as follows \cite[Chapter 9, Theorem 3.7]{HirschDiffTop}:
For a connected, compact, orientable surface $B$ of genus $g$ with $b$ boundary components,
there exists a diffeomorphism $B \to \Sigma_{g,b}$,
where $\Sigma_{g,b}$ is the $g$-holed torus with $b$ disjoint open discs removed.
When $B$ is also oriented, we may upgrade this classification to use diffeomorphisms
that preserve orientations. This is the form of the classification theorem
used in the proof of Proposition \ref{prop:diffeomorphic-bases}.
There always exists an orientation-reversing self-diffeomorphism
$r : \Sigma_{g,b} \to \Sigma_{g,b}$.
It is a trivial adaptation of the $b=0$ case \cite[Theorem A]{Muellner} to construct one.
Specifically, $\Sigma_{g,b}$ may be embedded in $\mathbb{R}^3$ reflection-symmetrically about a plane,
say $z=0$ in Cartesian coordinates.
Then reflection about the plane $(x,y,z) \mapsto (x,y,-z)$ gives $r$.
If the diffeomorphism $B \to  \Sigma_{g,b}$ established above does not preserve orientations,
we post-compose it with $r$ to obtain a new diffeomorphism which does.
\end{rem}

\begin{rem}
In Proposition \ref{prop:diffeomorphic-bases},
the hypothesis on topological ends of $M$ is precisely an $\mathbb{R}$-action-agnostic way
to impose that the quotient spaces $B_\pm$ and $B_\pm^\prime$ are
compact.\footnote{Compactness of the quotient space $B_\pm$ implies that
each component of $M \cong \mathbb{R} \times B_\pm$ has two distinct ends.
To see this in case of connected $M$, take the two decreasing sequences
$U^i_0 \supseteq U^i_1 \supseteq \cdots \supseteq U^i_n \supseteq \cdots$, for $i=1,2$,
of non-empty connected open sets given by $U^1_n := (n,\infty) \times B_\pm$ 
and $U^2_n := (-\infty, -n) \times B_\pm$ respectively.
These sets satisfy $\cap_{n \in \mathbb{N}} \overline{U^i_n} = \emptyset$
and have compact boundaries,
$\partial U^1_n = \{n\} \times B_\pm$ and $\partial U^2_n = \{-n\} \times B_\pm$.
Thus $\{ U^1_n\}_{n \in \mathbb{N}}$ and $\{ U^2_n \}_{n \in \mathbb{N}}$ represent ends of $M$.
They are distinct since $U^1_n \cap U^2_m = \emptyset$ for all $n,m \in \mathbb{N}$.}
This compactness allows us to address the cancellation problem
(i.e.\ show that $\mathbb{R} \times B_\pm \cong \mathbb{R} \times B_\pm^\prime$ implies $B_\pm \cong B_\pm^\prime$)
using the classical classification of compact surfaces.
There exist analogous classification results for non-compact surfaces \cite{BrownMesser}
but their application here would require a detailed analysis
of the ends of $B_\pm$ and $B_\pm^\prime$, which we avoid for simplicity.
\sk

Many spaces $M$ of interest satisfy the hypotheses in Proposition \ref{prop:diffeomorphic-bases},
e.g.\ the filled cylinder $\mathbb{R} \times \mathbb{D}^2$ (Example \ref{ex:cylinder})
and the annular cylinder $\mathbb{R} \times \mathbb{A}^2$ where $\mathbb{A}^2 :=
\{(x,y)\in\bbR^2\,\vert\,\frac{1}{2}\leq x^2+y^2 \leq 1\,\}$ is an annulus.
We note, however, that some important examples do not,
including the half-space (Example \ref{ex:Minkowski3}).
To complete the treatment of such an example,
one should either generalize Proposition \ref{prop:diffeomorphic-bases} appropriately,
or prove an analogous statement specific to the example of interest using ad hoc methods.
\end{rem}

Suppose that $f : B_\pm \to B^\prime_\pm$ is an orientation-preserving diffeomorphism.
Recall that the Poisson structure $\sigma_{(0)}$ of \eqref{eqn:sigmapoisson}
is defined using only wedge products and integration of forms,
and so it is compatible with pullbacks along orientation-preserving diffeomorphisms.
Hence, pulling forms back along $f$ defines a natural isomorphism
\begin{flalign} \label{eqn:diffeomorphic-bases:nat-iso}
f^\ast \, :\,
\big( \Omega^{\bullet\,\prime}_{\partial,c}[1], \sigma^\prime_{(0)} \big) \circ f \,
\stackrel{\cong}{\Longrightarrow} \,
\big( \Omega^\bullet_{\partial,c}[1], \sigma_{(0)} \big)
\end{flalign}
between functors from $\mathbf{Open}(B_\pm)$ to $\mathbf{PoCh}_\mathbb{R}$.
By abuse of notation,
we denote the functor $f : \mathbf{Open}(B_\pm) \to \mathbf{Open}(B_\pm^\prime)$ by
the same symbol as the diffeomorphism which induces it.
\begin{cor}\label{cor:independence}
Suppose that the hypothesis of Proposition \ref{prop:diffeomorphic-bases} holds
and denote by $f : B_\pm \to B^\prime_\pm$ the resulting orientation-preserving diffeomorphism.
There is a natural isomorphism
\begin{flalign}
\CCR(f^\ast) \, :\,
\mathfrak{B}^\prime_\pm \circ f \,
\stackrel{\cong}{\Longrightarrow} \,
\mathfrak{B}_\pm
\end{flalign}
between AQFTs on $\mathbf{Open}(B_\pm)$,
defined by applying the $\CCR$-functor to the natural isomorphism \eqref{eqn:diffeomorphic-bases:nat-iso}.
\end{cor}


\section{\label{sec:analysis}Analyzing the dimensionally reduced AQFT}
The aim of this section is to analyze some of the physics encoded in
the dimensionally reduced AQFT $\BBB_\pm : \Open(B_\pm)\to\dgastAlg$ 
from Corollary \ref{cor:dimreducedAQFT}. In particular, we shall show that, restricting this AQFT
to the interior $\mathrm{Int}B_\pm$, one recovers the usual linear Chern-Simons theory on surfaces
without boundary and, restricting to a tubular neighborhood of the boundary $\partial B_\pm$, 
one finds a chiral free boson.

\paragraph{Restriction to the interior $\mathrm{Int}B_\pm$:} Let us denote by 
$\Open(\mathrm{Int}B_\pm)\subseteq \Open(B_\pm)$
the full subcategory of non-empty open subsets $V\subseteq B_\pm$ that do not 
intersect the boundary, i.e.\ $V\cap \partial B_\pm = \emptyset$. On such regions
the boundary condition in $\Omega^\bullet_{\partial,\cc}(V)[1] $ becomes void. Hence, the restricted AQFT
$\BBB_\pm\vert_{\mathrm{Int}B_\pm} : \Open(\mathrm{Int}B_\pm)\to \dgastAlg$
assigns to $(V\subseteq \mathrm{Int}B_\pm)\in \Open(\mathrm{Int}B_\pm)$
the dg-algebra
\begin{flalign}
\BBB_\pm\vert_{\mathrm{Int}B_\pm}(V) \,=\,\CCR\big(\Omega_\cc^\bullet(V)[1],\sigma_{(0)}\big) \,\in\,\dgastAlg
\end{flalign}
given by canonical quantization of the $1$-shifted compactly supported
de Rham complex along the linear Poisson structure
\eqref{eqn:sigmapoisson}. Observe that the latter is (a homological refinement of) the 
Atiyah-Bott Poisson structure \cite{AtiyahBott} for flat connections on surfaces,
hence the restricted AQFT $\BBB_\pm\vert_{\mathrm{Int}B_\pm}$ can be interpreted
in terms of linear Chern-Simons theory on $\mathrm{Int} B_\pm$. A truncation
of this AQFT was previously studied in \cite{DMSChernSimons}.

\paragraph{Restriction near the boundary $\partial B_\pm$:} In order to study 
the induced theory on the boundary $\partial B_\pm$, we consider a smooth embedding 
$\partial B_\pm \times [0,1) \hookrightarrow B_\pm$ as a tubular neighborhood of the boundary 
$\partial B_\pm$, so its restriction to $\partial B_\pm \times \{ 0 \}$ embeds as $\partial B_\pm$.
\sk

Let us denote by $\Open(\partial B_\pm)$ the category 
whose objects are all non-empty open subsets $W\subseteq \partial B_\pm$ of the $1$-dimensional 
boundary manifold $\partial B_\pm$ and whose morphisms $\iota_{W}^{W^\prime} : W\to W^\prime$ are subset 
inclusions $W\subseteq W^\prime\subseteq \partial B_\pm$. We endow this category with the 
disjointness orthogonality relation, i.e.\ $(W_1\subseteq W)\perp (W_2\subseteq W)$ 
if and only if $W_1\cap W_2 =\emptyset$, and introduce the orthogonal functor
$\Open(\partial B_\pm)\to \Open(B_\pm)$ defined
by $ W \mapsto W\times [0,1) $.
We denote by $\BBB_\pm\vert_{\partial B_\pm} : \Open(\partial B_\pm) \to \dgastAlg$ the restriction
of the AQFT $\BBB_\pm$ from Corollary \ref{cor:dimreducedAQFT} along this orthogonal functor.
Explicitly, this theory assigns to $(W\subseteq \partial B_\pm)\in \Open(\partial B_\pm)$
the dg-algebra
\begin{flalign}
\BBB_\pm\vert_{\partial B_\pm}(W)\,=\, \CCR\big(\Omega_{\partial,\cc}^\bullet\big(W\times[0,1)\big)[1],\sigma_{(0)}\big) \,\in\,\dgastAlg\quad.
\end{flalign}
We will now construct a weakly equivalent model for the AQFT $\BBB_\pm\vert_{\partial B_\pm}$
that is determined only from $1$-dimensional data on $\partial B_\pm$.
\sk

For the following construction, we choose a global coordinate $r$ on $[0,1)$ 
and pick any compactly supported $1$-form $\omega \in \Omega^1_\cc\big([0,1)\big)$
with unit integral $\int_0^1 \omega=1$. As before, we simply write 
$\omega := \pr_{[0,1)}^\ast(\omega)\in \Omega^1\big(W\times[0,1)\big)$
for the form obtained via pullback along the projection map, and similar for pullbacks
of forms from $W$. We define, for each non-empty open subset $(W\subseteq \partial B_\pm)\in \Open(\partial B_\pm)$, the map
\begin{flalign}\label{eqn:kappacochainmap}
\kappa \,:\, \Omega^0_\cc(W) \,\longrightarrow\,\Omega_{\partial,\cc}^\bullet\big(W\times[0,1)\big)[1]~~,\quad
\varphi\,\longmapsto\, \omega\wedge\varphi - \int_r^1\omega \wedge \dd \varphi
\quad.
\end{flalign}
Regarding the domain of this map as a cochain complex concentrated in degree $0$ (with trivial differential),
one easily checks that $\kappa$ defines a cochain map. Indeed, 
\begin{flalign}
-\dd \kappa(\varphi) &= -\dd\bigg(\omega\wedge\varphi - \int_r^1\omega \wedge \dd \varphi\bigg)
=\omega\wedge\dd \varphi - \omega\wedge\dd \varphi = 0\quad,
\end{flalign}
for all $\varphi\in \Omega^0_\cc(W)$.
\begin{lem}\label{lem:boundarysimplification}
For each non-empty open subset $(W\subseteq \partial B_\pm)\in \Open(\partial B_\pm)$, 
the cochain map \eqref{eqn:kappacochainmap} is a quasi-isomorphism.
\end{lem}
\begin{proof}
The proof is similar to the one of \cite[Theorem 4.1]{GRW}, so we can be relatively brief.
A quasi-inverse for $\kappa$ is given by the following fiber integration
\begin{flalign}\label{eqn:kappaquasiinverse}
\parbox{0.5cm}{\xymatrix{
\Omega_{\partial,\cc}^\bullet\big(W\times[0,1)\big)[1] \ar[d]_-{\lambda}\\
\Omega^0_\cc(W)
}
}\!\!=~ \left(\parbox{2cm}{\xymatrix@C=2em{
\ar[d]^-{0}\Omega^0_{\partial,\cc}\big(W\times[0,1)\big)  \ar[r]^-{-\dd} & \ar[d]^-{\int_0^1} \Omega^1_{\cc}\big(W\times[0,1)\big) \ar[r]^-{-\dd} & \Omega^2_{\cc}\big(W\times[0,1)\big) \ar[d]^-{0}\\
0 \ar[r]_-{0} & \Omega^0_\cc(W) \ar[r]_-{0}&0
}}\right)
\; \; ,
\end{flalign}
which is a cochain map owing to the boundary condition on $\Omega^0_{\partial, \mathrm{c}}\big(W \times [0,1) \big)$.
One checks that $\lambda \circ \kappa = \id$ since $\omega$ integrates to $1$.
The other composition is homotopic to the identity, i.e.\ $\id- \kappa\circ \lambda = \partial K$,
with cochain homotopy $K\in \big[\Omega_{\partial,\cc}^\bullet\big(W\times[0,1)\big)[1] ,
\Omega_{\partial,\cc}^\bullet\big(W\times[0,1)\big)[1] \big]^{-1}$ defined by
\begin{flalign}\label{eqn:Khomotopy}
K(\alpha)
= \begin{cases}
0
&~,~~\text{for }\vert\alpha\vert = -1
\quad,\\[0.5em]
\int_r^1\alpha - \int_r^1\omega\wedge\int_0^1\alpha 
&~,~~\text{for }\vert\alpha\vert = 0
\quad,\\[0.5em]
\int_r^1 \alpha
&~,~~\text{for }\vert\alpha\vert = 1\quad.
\end{cases}
\end{flalign}
This cochain homotopy is compatible with the boundary condition
on $0$-forms since, for all $\alpha\in \Omega_{\partial,\cc}^\bullet\big(W\times[0,1)\big)[1]^0$,
we have that
$K(\alpha)\vert_{r=0}
= \int_0^1 \alpha - \int_0^1\omega\wedge\int_0^1 \alpha
=0$.
\end{proof}

\begin{rem}
Building on Remark \ref{rem:fieldcomment}, let us also describe the 
quasi-isomorphism $\kappa$ in \eqref{eqn:kappacochainmap} using the dual language of 
fields instead of observables. It turns out to be simpler to interpret 
the quasi-inverse cochain map $\lambda$ in \eqref{eqn:kappaquasiinverse}.
The $2$-dimensional fields are described by the boundary conditioned $1$-shifted
de Rham complex $\Omega_{\partial}^\bullet\big(W\times[0,1)\big)[1]$
and the $1$-dimensional fields by the complex $\Omega^1(W)$ that is concentrated
in degree $0$. (One should think of elements $J\in \Omega^1(W)$ as the Abelian 
currents of a chiral free boson.)
The dual of $\lambda$ is then given by the pullback map 
$\pr_{W}^\ast : \Omega^1(W)\to \Omega_{\partial}^\bullet\big(W\times[0,1)\big)[1]$
along the projection $\pr_W : W \times[0,1)\to W$. Thinking of $W$ as the boundary
$W\times\{0\}$ of $W \times[0,1)$, one can interpret this map as the constant 
evolution of boundary values in $\Omega^1(W)$ along the $[0,1)$-factor.
Hence, the origin of the quasi-isomorphism \eqref{eqn:kappaquasiinverse}, and hence of \eqref{eqn:kappacochainmap},
lies in the fact that Chern-Simons theory on $W \times[0,1)$ is locally constant (i.e.\ topological)
normal to the boundary.
\end{rem}

Pulling back the unshifted linear Poisson structures \eqref{eqn:sigmapoisson} along the quasi-isomorphisms
\eqref{eqn:kappacochainmap} defines a family of unshifted linear Poisson structures 
\begin{subequations}\label{eqn:upsilon}
\begin{flalign}
\upsilon_{(0)}^W := \sigma_{(0)}^{W\times[0,1)}\circ (\kappa\otimes\kappa)
\,:\,\Omega_\cc^0(W)\otimes \Omega_\cc^0(W) \,\longrightarrow\,\bbR\quad,
\end{flalign}
for all $(W\subseteq \partial B_\pm)\in \Open(\partial B_\pm)$.
With a short calculation, one computes explicitly the value
of $\upsilon_{(0)}^W$ and finds
\begin{flalign}
\nn \upsilon_{(0)}^W(\varphi\otimes\psi) &= 
-\bigg(\int_W\Big(\varphi\wedge\dd \psi - \psi\wedge\dd \varphi\Big)\bigg)~
\int_0^1 \omega \int_r^1\omega \\[4pt]
&=-\frac{1}{2}\int_W\Big(\varphi\wedge\dd \psi - \psi\wedge\dd \varphi\Big) =
-\int_W\varphi\wedge\dd \psi\quad,
\end{flalign}
\end{subequations}
for all $\varphi,\psi\in \Omega_\cc^0(W)$.
Note that $\big(\Omega^0_\cc(W),\upsilon_{(0)}^W\big)$
is the usual Poisson vector space for the chiral free boson on $W$.
The cochain maps \eqref{eqn:kappacochainmap} then define the components of 
a natural transformation
\begin{flalign}\label{eqn:kappa}
\kappa\,:\, \big(\Omega^0_\cc,\upsilon_{(0)}\big)\,\Longrightarrow\, \big(\Omega^\bullet_{\partial,\cc}[1],\sigma_{(0)}\big)\big\vert_{\partial B_\pm}
\end{flalign}
of functors from $\Open(\partial B_\pm)$ to $\mathbf{PoCh}_\bbR$. This allows us to 
obtain our desired weakly equivalent model for the AQFT $\BBB_\pm\vert_{\partial B_\pm}$. 
Explicitly, consider the functor 
\begin{flalign}\label{eqn:CCCpm}
\CCC_\pm\, :=\, \CCR\big(\Omega^0_\cc,\upsilon_{(0)}\big)\,:\,\mathbf{Open}(\partial B_\pm) \,\longrightarrow\, \dgastAlg 
\end{flalign} 
defined by post-composing the functor $(\Omega^0_\cc,\upsilon_{(0)}) : 
\mathbf{Open}(\partial B_\pm)\to\mathbf{PoCh}_\bbR$ with 
the canonical commutation relations functor
$\CCR : \mathbf{PoCh}_\bbR\to \dgastAlg$ from \cite[Section 5]{LinearYM}.
Note that $\CCC_\pm$ satisfies the Einstein causality axiom 
on $\mathbf{Open}(\partial B_\pm)$ because $(\Omega^0_\cc,\upsilon_{(0)})$ satisfies 
a classical analogue of the Einstein causality axiom.
Indeed, for each orthogonal (i.e.\ disjoint) pair 
$(W_1\subseteq W)\perp (W_2\subseteq W)$ in $\mathbf{Open}(\partial B_\pm)$, 
$\upsilon^W_{(0)}\circ\big(\mathrm{ext}_{W_1}^{W}\otimes\mathrm{ext}_{W_2}^W\big)=0$ 
vanishes manifestly, see \eqref{eqn:upsilon}. 
\begin{cor}\label{cor:chiralbosonQFT}
The morphism (i.e.\ natural transformation)
\begin{flalign}
\CCR(\kappa)\,:\,\CCC_\pm \,\stackrel{\sim}{\Longrightarrow}\, \BBB_\pm\vert_{\partial B_\pm}
\end{flalign}
of AQFTs on $\Open(\partial B_\pm )$, defined by applying the 
$\CCR$-functor to the natural transformation \eqref{eqn:kappa}, is a natural weak equivalence. Hence,
the AQFT $\CCC_\pm$ from \eqref{eqn:CCCpm}, which describes the chiral free boson, provides
a weakly equivalent model for the boundary restriction $\BBB_\pm\vert_{\partial B_\pm}$
of the dimensionally reduced bulk/boundary AQFT from Corollary \ref{cor:dimreducedAQFT}.
\end{cor}
\begin{proof}
This follows directly from Lemma \ref{lem:boundarysimplification} and 
the fact that the $\CCR$-functor preserves weak equivalences, see \cite[Proposition 5.3]{LinearYM}.
\end{proof}

\paragraph{Interplay between bulk and boundary:} The AQFT $\BBB_\pm : \mathbf{Open}(B_\pm)\to\dgastAlg$
from Corollary \ref{cor:dimreducedAQFT} encodes also the interplay between Chern-Simons
observables in the bulk and chiral free boson observables on the boundary.
For instance, given any open subset of the form $W\times [0,1) \subseteq B_\pm$ 
in the tubular neighborhood of the boundary, we obtain a morphism
\begin{flalign}
\iota \,:\, W\times \big(\tfrac{1}{2},1\big) ~\longrightarrow ~ W\times [0,1)
\end{flalign} 
in the category $\mathbf{Open}(B_\pm)$ whose source is an interior region in the bulk and 
whose target is a region that intersects the boundary. (The particular choice of the open interval
$(\tfrac{1}{2},1)\subseteq [0,1)$ is inessential and purely conventional. Since Chern-Simons theory
is a topological QFT in the bulk, we could choose equally well another open interval $(a,b)\in [0,1)$.)
Applying the functor 
$\BBB_\pm : \mathbf{Open}(B_\pm)\to\dgastAlg$ on this morphism and using the natural weak equivalence
from Corollary \ref{cor:chiralbosonQFT} yields the following zig-zag
\begin{flalign}\label{eqn:zigzag}
\xymatrix@C=3em{
\BBB_\pm\big(W\times \big(\tfrac{1}{2},1\big)\big)\ar[r]^-{\BBB_\pm(\iota)}~&~ 
\BBB_\pm\big(W\times [0,1)\big) ~&~\ar[l]^-{\sim}_-{\CCR(\kappa)}\CCC_\pm(W)
}
\end{flalign}
in $\dgastAlg$, where $\sim$ indicates weak equivalence. This zig-zag allows us to relate the
Chern-Simons observables $\BBB_\pm\big(W\times \big(\tfrac{1}{2},1\big)\big)$
in the interior region to the chiral free boson observables $\CCC_\pm(W)$ in the $1$-dimensional
boundary subset $W\subseteq \partial B_\pm$. 
\begin{ex}\label{ex:holonomy}
Let us illustrate this relationship between bulk and boundary observables
through an example. Consider the case where the boundary region $W =\bbS^1$ is a circle.
Linear Chern-Simons theory on $\bbS^1\times (\tfrac{1}{2},1)$ contains observables that
measure the holonomy of flat $\bbR$-connections along non-trivial cycles. A particular choice
of such observable is given by the linear observable
\begin{flalign}\label{eqn:omegaholonomy}
\omega \,:=\, \pr^\ast_{(\frac{1}{2},1)}(\omega) \,\in\,\Omega^1_\cc\big(\bbS^1\times \big(\tfrac{1}{2},1\big)\big) \,\subseteq \,
\BBB_\pm\big(\bbS^1\times \big(\tfrac{1}{2},1\big)\big)
\end{flalign}
that is obtained by pulling back along the projection 
$\pr_{(\frac{1}{2},1)}  : \bbS^1\times(\tfrac{1}{2},1) \to (\tfrac{1}{2},1)$
a compactly supported $1$-form $\omega\in\Omega^1_\cc\big((\tfrac{1}{2},1)\big)$ 
with unit integral $\int_{\frac{1}{2}}^1\omega=1$.
To see that this observable indeed measures the holonomy of flat connections, let us
consider the family of flat $\bbR$-connections 
$A = \alpha\,\dd \phi\in \Omega^1\big(\bbS^1\times (\tfrac{1}{2},1)\big)$, 
where $\alpha =\int_{\bbS^1}A\in\bbR$ characterizes the holonomy 
and $\phi$ is the angle coordinate on $\bbS^1$. Evaluating
the observable $\omega$ on $A$ gives
\begin{flalign}
\int_{\bbS^1\times (\frac{1}{2},1)} \omega\wedge A \,=\, - \bigg(\int_{\bbS^1} A\bigg) ~\bigg(\int_{\frac{1}{2}}^1\omega\bigg)\,=\,  - \int_{\bbS^1} A\,=\,-\alpha\quad,
\end{flalign}
i.e.\ the observable $\omega$ distinguishes between the different flat connections.
Applying the morphism $\BBB_\pm(\iota)$ from \eqref{eqn:zigzag} allows us to consider the observable 
$\omega$ as an element in $\Omega^1_\cc\big(\bbS^1\times [0,1)\big)$. 
Using also the quasi-inverse of $\kappa$, given by the cochain map 
$\lambda$ in \eqref{eqn:kappaquasiinverse}, we can assign to the latter observable
the element
\begin{flalign}
\lambda(\omega) = \int_0^1\omega = 1\,\in\,\Omega_\cc^0(\bbS^1)\subseteq \CCC_\pm(\bbS^1)\quad,
\end{flalign}
which by construction satisfies $\omega - \kappa\lambda(\omega) = -\dd K(\omega)$ 
for $K$ the cochain homotopy displayed in \eqref{eqn:Khomotopy}. In words, this means that
the linear Chern-Simons observable \eqref{eqn:omegaholonomy}, 
which measures the holonomy of flat connections, can be 
presented (up to exact terms) by the linear observable 
$\lambda(\omega) =1\in  \CCC_\pm(\bbS^1)$ of the chiral free boson, 
which measures its zero mode.
\end{ex}


\section*{Acknowledgments}
The work of M.B.\ is fostered by 
the National Group of Mathematical Physics (GNFM-INdAM (IT)). 
A.G-S.\ is supported by a Royal Society Enhancement Grant (RF\textbackslash ERE\textbackslash 210053).
A.S.\ gratefully acknowledges the support of 
the Royal Society (UK) through a Royal Society University 
Research Fellowship (URF\textbackslash R\textbackslash 211015)
and an Enhancement Grant (RF\textbackslash ERE\textbackslash 210053).

%

\appendix

\section{\label{app:appendix}Technical details}

\subsection{\label{app:proper-flows}Proper \texorpdfstring{$\mathbb{R}$}{R}-actions}
\begin{lem} \label{lem:proper-flows-are-free}
Any proper smooth $\mathbb{R}$-action
$\Theta : \mathbb{R} \times N \to N$
on a smooth manifold $N$ is also free.
\end{lem}
\begin{proof}
For each $p \in N$, the curve $\Theta(\,\cdot\,, p) : \mathbb{R} \to N$ 
is the integral curve through $p$ of the vector field generating the flow $\Theta$.
Assume the $\mathbb{R}$-action $\Theta$ is not free,
so there exists some $\Delta s \neq 0$ and $p \in N$ with $\Theta(\Delta s, p) = p$.
Then by uniqueness of integral curves, $\Theta(\,\cdot\, , p)$
is periodic, i.e.\ $\Theta(n \Delta s,p) = p$ for all $n \in \mathbb{Z}$.
It follows that $\mathbb{R}_{\{p\}}:=\{s \in \mathbb{R} \,\vert\, \Theta (s, \{p\}) \cap \{p\} \neq \emptyset\}
\supseteq \{ n \Delta s \,vert\, n \in \mathbb{Z} \}$
is not compact in $\mathbb{R}$. Hence, the $\mathbb{R}$-action $\Theta$ is not proper.
\end{proof}

\begin{cor} \label{cor:proper_flow_gives_principal_bundle}
For any proper smooth $\mathbb{R}$-action $\Theta : \mathbb{R} \times N \to N$ on a smooth manifold $N$, 
the quotient $\pi : N \to N / \mathbb{R}$ by $\Theta$ is a 
smooth principal $\mathbb{R}$-bundle with right $\mathbb{R}$-action $\Theta$.
\end{cor}

\begin{lem} \label{lem:principal_R-bundles_are_oriented}
The canonical orientation on $\mathbb{R}$ induces a canonical bundle orientation
on any principal $\mathbb{R}$-bundle $\pi : N \to B$.
\end{lem}
\begin{proof}
Recall that a bundle orientation on $\pi : N \to B$ can be expressed as
an orientation of the vector bundle $VN \subseteq TN$ of vertical vectors with respect to $\pi$,
see e.g.\ \cite[Section 7.4, Proposition II]{GHV}.
For any element $a \in \mathbb{R}$ in the Lie algebra of the structure group $\mathbb{R}$,
there is a fundamental vector field $A \in \Gamma^\infty(VN)$
on the total space $N$ of the principal bundle, see e.g.\ \cite[Section A.1]{Choquet-BruhatEtAl}.
These fundamental vector fields give $\mathbb{R}$-linear isomorphisms 
$\mathbb{R} \to V_pN\,,~a \mapsto A_p$ between the Lie algebra and the space of vertical vectors at any $p \in N$.
\sk
	
The canonical orientation on the Lie group $\mathbb{R}$ identifies 
bases $\{a\}$ of its Lie algebra $\mathbb{R}$ as positively oriented if $a > 0$.
The bundle orientation on $\pi : N \to B$ is then induced by transporting along 
the isomorphisms above, i.e.\ $\{A_p\}$ is positively oriented in $V_pN$ for $a > 0$. 
This is a smooth choice of positively oriented bases since the fundamental vector fields $A$ are smooth.
\end{proof}

Consider any fiber bundle $\pi : N \to B$ and denote by $n := \dim N$ and $b := \dim B$
the dimensions of, respectively, the total and base spaces.
Manifold orientations on $N$ and $B$ and a bundle orientation on $\pi: N \to B$ are called
\emph{compatible} (in the fiber-first convention) if,
for representative forms $\omega_{B} \in \Omega^{b}(B)$ and $\beta_\pi \in \Omega^{n-b}(N)$
of the base space and bundle orientations, the total space
orientation is represented by $\beta_\pi \wedge \pi^\ast \omega_{B}$.
\begin{cor} \label{cor:principal_R-bundles:total_base_orientations}
For any principal $\mathbb{R}$-bundle $\pi : N \to B$, 
compatibility with the canonical bundle orientation of Lemma \ref{lem:principal_R-bundles_are_oriented} 
defines a bijection between the set of orientations on $N$ and the set of orientations on $B$.
\end{cor}
\begin{proof}
Consider first the case where either set of orientations is empty,
i.e.\ either $N$ or $B$ is non-orientable. Since any principal $\mathbb{R}$-bundle is trivializable 
$N \cong \mathbb{R} \times B$ and $\mathbb{R}$ is orientable, $N$ is orientable if and only $B$ 
is orientable. Hence, in this case both sets of orientations are empty and one has the claimed bijection.
\sk

Consider now the case where both $N$ and $B$ are orientable.
It is trivial to produce a compatible orientation on $N$ from any orientation on $B$, 
using the definition of compatibility. For an inverse, when given an orientation on $N$, 
choose an arbitrary orientation on $B$ and transform it to a compatible one by reversing 
it on those components where it is incompatible.
\end{proof}

\subsection{\label{app:properties_of_Theta-pasts-futures}Properties of \texorpdfstring{$\Theta$}{Theta}-future/past sets}
The $\Theta$-future/past sets entering property \ref{def:Greenshomotopies:support} 
of Definition \ref{def:Greenshomotopies} may be expressed in terms of a binary 
relation $\leq_\Theta$ on $N$ associated to the $\bbR$-action $\Theta : \mathbb{R} \times N \to N$, defined as
\begin{flalign} \label{eqn:binary_relation_from_flow}
p \leq_\Theta q \quad\text{if}\quad \exists s \geq 0 \text{ such that } q = \Theta(s,p)\quad.
\end{flalign}
The $\Theta$-future/past $J^{\uparrow / \downarrow}_\Theta(S)$ of a subset
$S \subseteq N$ are respectively the successor- and predecessor-sets of $S$ 
with respect to $\leq_\Theta$. Compare this to Remark \ref{rem:ne/sw_relation}.
\begin{lem}
For any proper smooth $\mathbb{R}$-action $\Theta : \mathbb{R} \times N \to N$ on a smooth manifold $N$,
the (graph of the) relation $\leq_\Theta$ defined in \eqref{eqn:binary_relation_from_flow} 
is a closed subset of $N \times N$.
\end{lem}
\begin{proof}
The (graph of the) relation $\leq_\Theta \,\subseteq N \times N$ of 
\eqref{eqn:binary_relation_from_flow} may be expressed as the image
\begin{flalign}
\leq_\Theta  \;= \theta \left([0,\infty) \times N\right) \subseteq N \times N
\end{flalign}
of a closed set under the shear map $\theta : \mathbb{R} \times N \to N \times N$ 
associated to $\mathbb{R}$-action $\Theta$ as per \eqref{eqn:shear_map}.
By hypothesis, $\Theta$ is a proper $\mathbb{R}$-action i.e.\ $\theta$ is a proper map.
Since $N \times N$ is a manifold, so Hausdorff and locally compact,
it follows by \cite[Theorem A.57]{LeeSmooth} that $\theta$ is a closed map.
Hence the relation $\leq_\Theta$ is a closed subset in $N \times N$.
\end{proof}

The following is a general fact about binary relations on topological spaces.
\begin{lem} \label{lem:closed_relations}
Let $X$ be a topological space and $R \subseteq X \times X$ a binary relation on $X$.
For any subset $S \subseteq X$, denote the successor-set of $S$ as
\begin{subequations}
\begin{flalign}
R^\uparrow(S)\, :=\, \left\{ x \in X \;\middle\vert\; (y,x) \in R \text{ for some } y \in S \right\}
\end{flalign}
and the predecessor-set of $S$ as
\begin{flalign}
R^\downarrow(S) \, :=\, \left\{ x \in X \;\middle\vert\; (x,y) \in R \text{ for some } y \in S \right\}\quad.
\end{flalign}
\end{subequations}
If $R$ is a closed subset of the product space $X \times X$, 
then $R^{\uparrow}(K)$ and $R^{\downarrow}(K)$ are closed for $K \subseteq X$ compact.
\end{lem}

This is proven for partial orders $R$ in \cite[Chapter 1, Proposition 4]{Nachbin}.
A proof applicable to general binary relations is given for 
\cite[Theorem 4.12]{Minguzzi}, which treats the special case of the causal 
relation on a Lorentzian manifold.
\begin{cor} \label{cor:Theta_future_past_of_compact_is_closed}
For any proper smooth $\mathbb{R}$-action $\Theta : \mathbb{R} \times N \to N$ on a smooth manifold $N$,
the $\Theta$-future/past sets $J^{\uparrow / \downarrow}_\Theta(K)$ of 
\eqref{eqn:Theta-past-future-sets} are closed for $K \subseteq N$ compact.
\end{cor}

\begin{lem} \label{lem:intersect-future-compact-past-Cauchy}
Let $N$ be a smooth manifold with a proper $\mathbb{R}$-action $\Theta : \mathbb{R} \times N \to N$.
Let further $K\subseteq N$ be a compact subset and $\sigma$ be a smooth section of the 
quotient $\pi : N \to N / \mathbb{R}$ by $\Theta$. Denote by $\Sigma$ the image of $\sigma$.
Then the subsets $J^\uparrow_\Theta(K) \cap J^\downarrow_\Theta (\Sigma)$ and 
$J^\downarrow_\Theta(K) \cap J^\uparrow_\Theta(\Sigma)$ of $N$ are both compact.
\end{lem}
\begin{proof}
By Corollary \ref{cor:proper_flow_gives_principal_bundle}, 
$\pi : N \to N / \mathbb{R}$ is a smooth principal $\mathbb{R}$-bundle.
Since smooth sections are smooth embeddings, its base $N / \mathbb{R}$ 
is diffeomorphic to $\Sigma$. Thus the trivialization induced by the section $\sigma$ 
gives a diffeomorphism $\Theta_\Sigma : \mathbb{R} \times \Sigma \xrightarrow{\sim} N\,,~(t,p) \mapsto \Theta(t,p)$.
\sk

We see that $J^\uparrow_\Theta(\Sigma) = \Theta_\Sigma \left([0,\infty) \times \Sigma\right)$
and $J^\downarrow_\Theta(\Sigma) = \Theta_\Sigma \left((-\infty, 0] \times \Sigma\right)$
are closed subsets in $N$,
so $J^{\uparrow / \downarrow}_\Theta(K) \cap J^{\downarrow / \uparrow}_\Theta (\Sigma)$ are 
also closed in $N$ because $J^{\uparrow / \downarrow}_\Theta(K)$ are closed by 
Corollary \ref{cor:Theta_future_past_of_compact_is_closed}.
\sk
	
With projections denoted $\pr_\mathbb{R} : \mathbb{R} \times \Sigma \to \mathbb{R}$ 
and $\pr_\Sigma : \mathbb{R} \times \Sigma \to \Sigma$, define 
$x := \inf \pr_\mathbb{R} \Theta_\Sigma^{-1} (K)$ which is finite because $K$ is compact.
Then $J^\uparrow_\Theta (K) \subseteq \Theta_\Sigma \left([x, \infty) \times \pr_\Sigma \Theta_\Sigma^{-1} (K) \right)$,
so that the closed set $J^\uparrow_\Theta(K) \cap J^\downarrow_\Theta(\Sigma)$ is contained in the 
compact set $\Theta_\Sigma \left([x,0] \times \pr_\Sigma \Theta_\Sigma^{-1} (K)\right)$ and is thus also compact.
Similarly, $J^\downarrow_\Theta(K) \cap J^\uparrow_\Theta(\Sigma)$ is compact since it is closed and 
contained in the compact set $\Theta_\Sigma \left([0,y] \times \pr_\Sigma \Theta_\Sigma^{-1} (K)\right)$, 
where $y := \sup \pr_\mathbb{R} \Theta_\Sigma^{-1} (K)$.
\end{proof}


\printbibliography

\end{document}